\documentclass[journal]{IEEEtran}
\IEEEoverridecommandlockouts
\usepackage{psfrag}
\usepackage{dsfont}
\usepackage{subcaption}
\usepackage{amsmath, amsthm, amssymb, amsfonts}
\usepackage{graphicx}

\usepackage{algorithm}
\usepackage{algorithmic}
\usepackage{epstopdf}
\newcommand{\bd}{\mathbf}
\newcommand{\xmi}{\bd{x_{-i}}}
\newcommand{\E}{\mathbb{E}}
\newcommand{\ymax}{y_{\max}}
\newcommand{\ox}{\overline{x}}
\newcommand{\opi}{\overline{\pi}}
\newcommand{\mS}{\mathcal{S}}
\newcommand{\Del}{\Delta}
\newcommand{\newbz}{}
\newcommand{\newrj}{}
\newcommand{\hymax}{\hat{y}_{\max}}
\DeclareMathOperator{\Cov}{cov}

\newtheorem{assumption}{Assumption}
\newtheorem{proposition}{Proposition}
\newtheorem{remark}{Remark}
\newtheorem{lemma}{Lemma}
\newtheorem{corollary}{Corollary}
\newtheorem{definition}{Definition}
\newtheorem{theorem}{Theorem}
\newtheorem{example}{Example}
\begin{document}
%%%%%%%%%%%%%%%%

\title{Competition and Efficiency of Coalitions in Cournot Games with Uncertainty }

\author{Baosen Zhang,~\IEEEmembership{Member,~IEEE,}
        Ramesh Johari,~\IEEEmembership{Member,~IEEE,}
        Ram Rajagopal,~\IEEEmembership{Member,~IEEE,}% <-this % stops a space
\thanks{B. Zhang is with the Department of Electrical Engineering at the University of Washington, email: zhangbao@uw.edu.}% <-this % stops a space
\thanks{R. Johari is with the Department of Management \& Science and Engineering at Stanford University, email: ramesh.johari@stanford.edu.}%
\thanks{R. Rajagopal is with the Department of Civil and Environmental Engineering at Stanford University, email: ramr@stanford.edu.}% <-this % stops a space
}

\maketitle
\begin{abstract}
We investigate the impact of coalition formation on the efficiency of Cournot games where producers face uncertainties. In particular, we study a market model where firms must determine their output before an uncertain production capacity is realized. In contrast to standard Cournot models, we show that the game is not efficient when there are many small firms. Instead, producers tend to act conservatively to hedge against their risks. We show that in the presence of uncertainty, the game becomes efficient when firms are allowed to take advantage of diversity to form groups of certain sizes. We characterize the tradeoff between market power and uncertainty reduction as a function of group size.  \newrj{In particular, we compare the welfare and output obtained with coalitional competition, with the same benchmarks when output is controlled by a single system operator.}
%We consider two closely related metrics to measure efficiency of coalitions: \newbz{welfare and allocation}.  The former is the total welfare when output is controlled by a \newbz{system operator}, while the latter is the total production of goods under a \newbz{system operator}.
We show when there are $N$ firms present, competition between groups of size $\Omega(\sqrt{N})$ results in equilibria that are socially optimal in terms of welfare and groups of size $\Omega(N^{2/3})$ are socially optimal in terms of production. We also extend our results to the case of uncertain demand by establishing an equivalency between Cournot oligopoly and Cournot Oligopsony. We demonstrate our results with real data from electricity markets with significant wind power penetration.
\end{abstract}%

%\KEYWORDS{deterministic inventory theory; infinite linear programming duality;
%  existence of optimal policies; semi-Markov decision process; cyclic schedule}

\begin{IEEEkeywords}
  Cournot Games, Decision Making under Uncertainty, Efficiency, Coalitions, Electricity Markets
\end{IEEEkeywords}
%\History{This work is partially supported by the Stanford Precourt Institute for Energy Efficiency and the US-Israel Binational Foundation.}

%%%%%%%%%%%%%%%%%%%%%%%%%%%%%%%%%%%%%%%%%%%%%%%%%%%%%%%%%%%%%%%%%%%%%%

\section{Introduction}
%\todo{Add symmetric assumption to X, add X has finite variance, variance is $\sigma^2$, change abstract, check if theorem statements are consistent, weaken condition in lemma 1}

{\em Cournot games} are among the most extensively studied models for oligopolistic competition among multiple firms.  A Cournot oligopoly is a model where participants compete with each other by controlling the amount of a homogeneous good that they produce.  A market price is determined as a function of the total output of the firms.  The profit of a firm is then the product of the market price and their output quantity, less any costs incurred.  The producers are assumed to act strategically and rationally to maximize their individual profits.   This model was initially studied by \cite{Cournot}; for surveys of such models, see, e.g., \cite{Shaprio89,Friedman77,Daughety88}.

In this paper, we consider Cournot competition among firms who face {\em production uncertainty}.  In the model we consider, firms first commit to an expected level of output; subsequently, actual production is realized, drawn from a distribution parameterized by the firm's {\em ex ante} chosen level.  Any shortfall from the precommitted level incurs a penalty.  Such a model captures production decisions by firms in environments where commitments must be made before all relevant factors influencing production are known.

Electricity markets serve as one motivating example of such an environment.  In electricity markets, producers submit their bids before the targeted time of delivery (e.g., one day ahead). However, renewable resources such as wind and solar have significant uncertainty (even on a day-ahead timescale).  As a result, producers face uncertainties about their actual production capacity at the commitment stage.

Our paper focuses on a fundamental tradeoff revealed in such games.  On one hand, in the classical Cournot model, efficiency obtains as the number of individual firms approaches infinity, as this weakens each firm's {\em market power} (ability to influence the market price through their production choice).  On the other hand, this result does not carryover when production uncertainty is present: firms protect themselves against the risk of being unable to meet the prior commitment by under-producing relative to the efficient level.

In considering how to recover efficient performance, we are naturally led to think of {\em coalitions} of firms.  Informally, if firms pool together, they can mitigate individual uncertainty any one of them may perceive in future production (a law of large numbers effect).  Of course, coalitions are not without their downside: coalitions possess greater market power than individual firms.  Indeed, this concern is substantial, as coalitions must be of fairly substantial size to mitigate the adverse effects of production uncertainty.  In the context of electricity markets, regulators must grapple with the consequences of allowing wind generators and other renewable resources to form coalitions as they bid into the market.  As a result we are led to a fundamental question: how many coalitions should be allowed to form, and of what size, if the regulator is interesting in maximizing overall market efficiency?

We characterize this tradeoff by studying the efficiency of Cournot competition when producers are allowed to form coalitions.  Our main contributions are as follows.  {\em First}, in the model described above, we characterize equilibrium among competing coalitions, as well as the socially optimal benchmark.  {\em Second}, as measures of efficiency, we compare both the welfare and production output of the firms under Cournot competition with the socially optimal welfare and output.  We characterize an optimal scaling regime for coalition structure (in the limit of many firms) under which the efficiency losses can be made arbitrarily small.  That is, there exist coalitions (partition of producers) that achieve essentially efficient reduction in uncertainty, but have no appreciable market power.  We also characterize the rate at which efficiency loss vanishes, and establish these results when firms may have correlated uncertainty.  Finally, by establishing an equivalence between Cournot oligopoly and Cournot oligopsony, we demonstrate our results are applicable to settings with demand uncertainty as well; \newrj{in particular, we illustrate the latter results in an application to urban parking allocation.}% \newbz{We illustrate the results on Cournot oligopsony via an example about urban parking in Seattle.}

Efficiency and welfare loss in Cournot games have been studied extensively in various contexts.  Early empirical analysis of welfare loss  was performed by \cite{Harberger54} and \cite{Bergson73}. Analytically, at the limit where many firms compete, many authors showed that a \emph{competitive limit} exists \cite{Frank63,Ruffin71,Haurie85,Novshek78}. The quantification of such a limit was considered by \cite{Anderson03} where the marginal costs of the firms are assumed to be constant. The work of \cite{Johari05} showed that for $N$ producers with the same cost function competing for a resource with a differentiable demand curve, the efficiency loss is no more than $1/(2N+1)$ when the producers are strategic and price anticipatory. The paper by \cite{Tsitsiklis13} derived a more general bound for convex demand curves and \cite{Guo05} studies the how the loss can be estimated in practice. The loss under asymmetric firms was studied by \cite{Corchon08}.

Most of the preceding literature concludes that full efficiency is achieved when a large number of producers are competing against each other. In this paper, we show that this is not the case when production uncertainty plays a role in the firms' profits.  Profit maximization under supply (or demand) uncertainty falls under the well studied {\em newsvendor problem} in the operations literature. However, most of previous work on this area is concerned with a single retailer \cite{Stevenson09}.  Oligopolistic competition is studied in \cite{Yao08} for additive demand, and in \cite{Adida11} for multiplicative demand; a related model with revenue sharing between different firms is discussed in \cite{Dana01}. To our knowledge, none of the previous works in this area consider efficient coalition formation. Another related research is area is contract designs~(see, e.g.  \cite{McAfee91}), where designers impose penalties to ensure that firms operate as expected. In this paper, the penalty is from uncertainties that are intrinsic to the problem.

One closely related work to our own is \cite{Yosha93}, where the author studies the role of intermediaries between diversification and competition in a large economy. Their results are derived under the assumption of a common randomness affecting all consumers, whereas in our work each producer faces its own randomness (possibly correlated with others).  This latter effect is what creates efficiency gains by allowing coalitions to form.

The remainder of the paper is organized as follows. Section \ref{sec:model} presents our basic model of Cournot competition with two stages: production commitments are made in the first stage, and actual available production capacity is realized in the second stage.  The market price is determined in the first stage based on production commitments.  In the second stage, the producer is charged a penalty if capacity is short of the firm's commitment level.  We then study the same model assuming firms act in coalitions.  As the focus of the paper is on the competition between groups, we do not study the nature of revenue sharing contracts within a group; see, e.g., \cite{Telser94} and \cite{Shapley67}.

In Section \ref{sec:iid}, we begin by studying the case where firms face i.i.d.~production uncertainty.  We first show that as the number of firms $N$ grows, the efficiency loss does not vanish (due to the adverse effects of production uncertainty).  We also show that the other extreme, a grand coalition of all producers, is inefficient (due to excessive exercise of market power).   We then study coalitional competition: in particular, we characterize the optimal group size and the optimal rate at which the efficiency loss approaches zero. By balancing the adverse effects of market power with the benefits of reduction in production uncertainty, we show that a coalition size of $\sqrt{N}$ producers (so $\sqrt{N}$ coalitions compete in the market) is optimal, and the efficiency loss is no more than $O(1/\sqrt{N})$. We show that same results hold under two models of correlated production uncertainty in Section \ref{sec:corr}.

Section \ref{sec:oligopsony} shows how the results can be applied to demand uncertainty by deriving an equivalence between Cournot oligopoly (firms competing to supply a good) and Cournot oligopsony (firms competing to consume a good). Section \ref{sec:power} illustrates how the results can be applied in practice with a case study of the electricity market in the Mid-Atlantic region of the United States when wind power producers are included in the market; this section is based on material in \cite{Zhang15}. Section \ref{sec:con} concludes the paper.

\section{Technical Preliminaries} \label{sec:model}
In this section, we define a two-stage game where multiple firms compete to satisfy the demand for a single resource. The main difference between the two-stage model and classical Cournot competition is that the bidding occurs in the first stage, but each producer has an uncertain production capacity at the time of delivery (second stage).

Suppose there are $N$ firms. The market operates in two stages as shown in Fig. \ref{fig:two-stage}.
\begin{figure}[ht]
\centering
\includegraphics[scale=0.8]{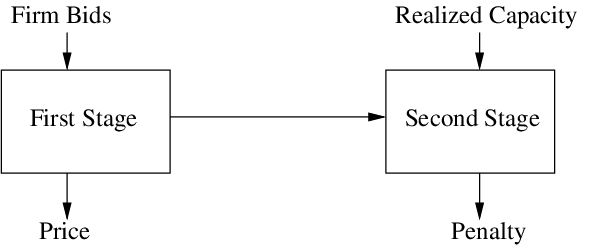}
\caption{The two-stage market model. Firms commit their production in the first stage, and this determines the price of the good. Capacities are realized in the second stage, and penalties are assessed if a firm's commitment is less than its realized capacity.}
\label{fig:two-stage}
\end{figure}
At the first stage, firm $i$ chooses a committed production quantity $x_i$ as its bid into the market. Let $p(y)$ denote the market price when $y$ units of aggregate output is committed. Let $X_i$ denote the capacity constraint on firm $i$'s production. Note $X_i$ is a random variable at the first stage, and is realized in the second stage. Throughout the paper, we will assume $X_i$'s are {\em continuous}, i.e., they follow a distribution with a continuous probability density function.  To focus on the effect of production uncertainty, we assume that each firm does not have a cost for producing the resource.\footnote{Our results would remain unchanged if each firm faced the same constant marginal production cost.}

If the promised amount ($x_i$) is larger than the capacity ($X_i$) firm $i$ is penalized by a cost $q$ per unit short fall.  Thus the cost of shortfall for firm $i$ is  $q(x_i-X_i)^+$.\footnote{For a real number $z$, let $z^+$ denote the positive part of $z$, i.e., $z^+ = z$ for $z> 0$, and $0$ otherwise.}  Without loss of generality, we set $q = 1$ for the remainder of the paper.

The assumption of a penalty linear in the shortfall is relatively common in the literature on newsvendor problems.  This penalty allows us to capture the {\em risk} associated with a shortfall, i.e., promising more than what can actually be delivered; indeed, the penalty serves to make a firm risk averse in its choice of commitment level.  Our main results continue to hold even for penalties of the form $\E[f(x_i - X_i)^+]$, as long as $f$ is convex, increasing, and has bounded derivative.
% $f$ that satisfy some additional smoothness constraints.

The structure of the penalty makes two important assumptions: first, that the penalty depends only on a firm's {\em own} shortfall (i.e., no inter-firm externalities); and that any excess production capacity cannot be resold in a secondary market.  In practice, these assumptions may be violated.  For example, in electricity markets, a real-time market is run to balance the realized supply and demand, and the study of such markets remain an important future direction for us.  Nevertheless, we believe our model captures the first order effects of production uncertainty on firm behavior, and on the role coalitions play in achieving efficient outcomes.  %This assumption allows us to focus on the effect of the second stage penalty on the strategy of the firms. It is also realistic in real world applications. For example, for wind farms, they do have minimal marginal cost for producing power.

%% \subsection{Second Stage Penalty}
%% In \eqref{eqn:pi}, we have chosen the second stage penalty to have the form of $\E[(x_i-X_i)^+]$. This penalty allows us to capture the \emph{risk} associated with a real-time shortfall, that is, promising more than what can actually be delivered. We show later in Section \ref{sec:iid} that the main results hold for penalties of the form $\E[f(x_i-X_i)^+]$, where $f$ is convex and satisfy some smoothness constraints. However, any risk of the form $f((x_i-X_i)^+)$ makes two implicit assumptions: the penalty amount only depends on a firms own shortfall and any excess cannot be sold back in some market. In some markets, there would be an opportunity for the firms to trade in real-time. Studying these type of markets is an important future direction.

We use the notation $\xmi$ to denote the quantities of chosen by all firms except $i$; that is, $\xmi=(x_1,x_2,\dots,x_{i-1},x_{i+1},\dots,x_N)$. The expected profit for firm $i$ is
\begin{equation} \label{eqn:pi}
\pi_i(x_i,\xmi)=p\left(\sum_{l=1}^N x_l\right) x_i- \E[(x_i-X_i)^+].
\end{equation}
%The term $q \E[(x_i-X_i)^+]$ represents the penalty a firm faces if its promised quantity $x_i$ is less than its capacity $X_i$. Here, we assume that the penalty is proportional to the shortfall, which is similar to the penalty in news vendor problems. This simple form is enough to allow us to study the actions of firms in a Cournot game when they face uncertainties. In practice, the exact form the penalty would depend on the target application. In some settings, a real-time market would be run to balance the realized supply and demand, and the study of such markets remain an important future direction for us. The parameter $q$ denotes the penalty price per unit of good, and without loss of generality, we take it to be $1$ going forward.

When each firm is price anticipatory, given $\xmi$, firm $i$ chooses $x_i>0$ to maximize $\pi_i$. A \emph{Nash equilibrium} of the game defined by $(\pi_1,\dots,\pi_N)$ is a vector $\bd x \geq 0$ such that for all $i$:
\begin{equation} \label{eqn:nash}
\pi_i(x_i,\xmi) \geq \pi_i(\tilde{x_i},\xmi), \text{   for all } \tilde{x_i} \geq 0.
\end{equation}
To analyze the Nash equilibrium for this game, we make the following assumptions on the price function $p$; this assumption remains in force for the entire paper.
\begin{assumption} \label{assump:p}
We assume that:
\begin{enumerate}
\item $p$ is strictly decreasing and $p(0) >0$ ;
\item $p(y)$ is concave and differentiable on $y \geq 0$ with $p'(0^+)<0$;
\item $p(y) \rightarrow - \infty$ as $y \rightarrow \infty $.
\end{enumerate}
\end{assumption}
Since $p$ is decreasing, $p(0) >0$, and tends to $-\infty$, there is a unique zero crossing point $\ymax$ such that:
\begin{equation}
\label{eq:ymax}
p(\ymax)=0.
\end{equation}

These assumptions are common in the literature (e.g., see \cite{Johari05}). The first assumption states that the price decreases as quantity increases and $p(0)>0$ avoids trivial solutions.  Concavity of the demand function is largely made for analytical convenience and to avoid long derivations.  (See Remark \ref{rem:concavity} at the end of this section, which suggests that key results on scaling of optimal coalition size continue to hold even with weaker assumptions on demand, e.g., logconcavity.)  The last assumption is also made for analytical simplicity.  In practice, the price becomes zero for large enough $y$. This is analytically undesirable since $p$ may not be globally concave, so in the third assumption we allow $p$ to be negative. This assumption is essentially without loss of generality since the regime of interest is always restricted to aggregate production where $p$ is non-negative (see Proposition \ref{prop:NE}).

We make the following assumptions on the random variable $X_i$ throughout the paper:
\begin{assumption} \label{assump:X}
For all $i$, $X_i$ is a continuous random variable with finite mean.
\end{assumption}
This assumption is made mainly for analytic convenience.
%Since $X_i$'s represents quantity of goods, we restrict $X_i$ to have nonzero probability of being positive to rule out trivial solutions.

It is now straight forward to show that a unique Nash equilibrium exists for the game $(\pi_1,\dots,\pi_N)$ as given in the following result.
\begin{proposition} \label{prop:NE}
Suppose $p$ satisfies Assumption \ref{assump:p} and the $X_i$'s satisfy Assumption \ref{assump:X}.
Then there exists a unique Nash equilibrium $\bd x$ for the game defined by $(\pi_1,\dots,\pi_N)$. Furthermore, $\sum x_i \leq y_{\max}$.
\end{proposition}
The proof of this proposition is given in Appendix \ref{app:NE}.

\begin{remark}
\label{rem:concavity}
The concavity assumption on the inverse-demand function $p$ can be relaxed to much weaker conditions. For example, suppose that $p$ is decreasing, continuous, and there exist a $\ymax$ such that $p(\ymax)=0$. Note $p$ is not necessarily concave. If all of the firms experience uncertainties with the same marginal distribution, the penalty function $\E[(x-X_i)^+]$ is identical for each firm $i$. By the result of \cite{Macmanus64} (or see Theorem 1 in \cite{Novshek85}), there exists a symmetric equilibrium for the game defined by $(\pi_1,\dots,\pi_N)$.

Of course without knowing more about $p$, it is hard to characterize the equilibria of the game. In this paper, we investigate the case of concave $p$, but as a consequence of the preceding result, we expect that our main conclusions should extend to more general types of inverse demand functions, at least in the case when firms are identical.
\end{remark}

\subsection{\newrj{The System Operator}}

We are interested in the efficiency of the Nash equilibrium of the game $(\pi_1,\dots,\pi_N)$.  \newrj{In this section, we outline the benchmark welfare we consider; in particular, we focus on the maximum welfare achieved by a centralized {\em system operator}.}
%In this section, we formally define the efficiency metrics we consider.  We focus in particular on , by studying the maximum welfare under an \newbz{operator that maximize the welfare of the overall system.}

To begin, given the price function $p$, we define aggregate consumer surplus in the usual way as:
\begin{equation}\label{eqn:U}
U(y)=\int_0^{y} p(z) dz.
\end{equation}

\newrj{The system operator aims to maximize aggregate consumer surplus, but incurs a cost a cost based on the \emph{expected aggregated shortfall}, assuming that it can control the output of all the producers.
Note that a key difference between this work and the existing literature is how the cost is accounted for: the system operator is effectively optimizing as if it is a ``grand coalition'' of all firms, with utility defined by aggregate consumer surplus.  This model is inspired by the role of the independent system operator (ISO) in electricity markets; see, e.g., \cite{Kirschen04}.}
%In this paper, we do not model the aggregate cost as the sum of cost of individual firms. Instead, \newbz{system operator faces  The operator's maximization problem is:}
\begin{subequations} \label{eqn:U_max}
\begin{align}
\mbox{maximize } & U \left(\sum_{i}^N x_i\right) - \E\left[\left(\sum_{i=1}^N x_i - \sum_{i=1}^N X_i \right)^+ \right] \\
\mbox{subject to } & x_i \geq 0, \; \forall \; i.
\end{align}
\end{subequations}
 The reason a system operator faces the aggregate shortfall, $\E\left[\left(\sum_{i=1}^N x_i - \sum_{i=1}^N X_i \right)^+ \right]$, instead of the sum of individual shortfalls, $\sum_{i=1}^N \E\left[\left( x_i -  X_i \right)^+ \right]$, is that it is always beneficial to equalize possible shortfalls in quantity by sharing resources. In the motivating example of electricity markets, suppose a system operator had control of all wind farms. Then as long as the sum of total realized wind, $\sum_i X_i$, is larger than the sum of the total committed wind, $\sum_i x_i$, there would be no additional cost.

The following lemma formalizes the benefit of aggregation to a system operator.
\begin{lemma}\label{lem:aggregate}
Suppose the $X_i$'s satisfy Assumption \ref{assump:X}. Then
\begin{equation}\label{eqn:aggregate_cost}
 \E\left[\left(\sum_{i=1}^N x_i - \sum_{i=1}^N X_i \right)^+ \right] \leq  \sum_{i=1}^N \E \left[ \left( x_i-X_i \right)^+ \right].
\end{equation}
\end{lemma}
The proof of this lemma is a straightforward application of Jensen's inequality and is given in Appendix \ref{app:aggregate}.

%In the above formulation, the social planner controls the aggregate output of all the firms, but still faces the aggregate uncertainty in $X_i$'s.  Note that we allow the social planner to {\em use the production capacity of one firm to offset the shortfall of another firm}.

Note that the objective function in \eqref{eqn:U_max} only depends on $\sum_{i=1}^N x_i$. With a change of variables, we can rewrite \eqref{eqn:U_max} as:
\begin{subequations} \label{eqn:U_max_y}
\begin{align}
\mbox{maximize } & U (y) - \E\left[\left(y - \sum_{i=1}^N X_i \right)^+ \right] \\
\mbox{subject to } & y \geq 0.
\end{align}
\end{subequations}
By Assumption \ref{assump:p}, $U$ is differentiable and concave, so the optimal solution to \eqref{eqn:U_max_y} is the unique positive solution to
\begin{equation}
\label{eq:ymaxprime}
p(y)-\Pr\left(y\geq \sum_{i=1}^N X_i\right) =0.
\end{equation}
%Uniqueness follows because, under our assumptions, the left hand side is strictly decreasing in $y$.

Thus we have the following corollary.
\begin{corollary}
\label{cor:eff}
A vector $x$ is efficient (i.e., solves the optimization problem \eqref{eqn:U_max}) if and only if:
\begin{equation}
\label{eq:efficiency}
\sum_i x_i = \ymax',
\end{equation}
where $\ymax'$ is the unique solution to \eqref{eq:ymaxprime}.  Further, $\ymax' \leq \ymax$.
\end{corollary}
Therefore at equilibrium, if the aggregate production of firms is $\ymax'$, the equilibrium is socially optimal.\footnote{It is straightforward to show, using the fact that $p$ is decreasing, that in any equilibrium firms produce {\em less} than $\ymax'$.  We omit this standard argument.}

\subsection{Efficiency Ratio}
In this section we define two closely related efficiency ratios. The first is the usual notion based on the gap between welfares of the Nash equilibrium and the \newrj{system} operator's optimal outcome.  \newrj{This is similar to the ``price of anarchy'' (e.g., \cite{koutsoupias1999worst}), but with the change that the benchmark is the system operator's payoff.} The second efficiency ratio is defined as the gap between the \emph{aggregate output} at the Nash equilibrium and $\ymax'$.

\begin{definition} \label{defn:r_u}
Consider the Cournot game $(\pi_1,\dots,\pi_N)$. Let $(x_1,\dots,x_N)$ be the Nash equilibrium of this game and let $\ymax'$ be the solution to \eqref{eq:ymaxprime}. The efficiency ratio $r_W$, with respect to the welfare, is:
\begin{equation} \label{eqn:r_u}
r_W= \frac{U\left(\sum_{i=1}^N x_i\right) - \sum_{i=1}^N  \E\left[ (x_i-X_i)^+ \right] }{U(\ymax')-\E \left[\left( \ymax'-\sum_i^N X_i \right)^+\right]}.
\end{equation}
The efficiency ratio $r_O$, with respect to the output, is:
\begin{equation} \label{eqn:r}
r_O=  \frac{\sum_{i=1}^N x_i}{\ymax'}.
\end{equation}
\end{definition}
Both efficiency ratios are of interest. The ratio of welfares measures how (in)efficient the Nash equilibrium is, and the ratio of quantities  directly compares the total output under the Nash equilibrium and the aggregate social output. For example, in electricity markets, the independent system operator (ISO) is typically interested in ensuing the maximum amount of renewable power (e.g., wind) is injected into the market. Therefore, $r_O$, the ratio of the total output under the Nash equilibrium and the maximum possible output $\ymax'$ can be a more direct measure of interest than $r_W$, the ratio of welfares.

The rest of this paper investigates the behavior of $r_W$ and $r_O$ as the number of firms and the number of coalitions grow. In particular, we characterize the asymptotic scaling rate of both efficiency ratios.

\subsection{Deterministic Cournot Games} \label{sec:deter}

Before moving on to the main result in Sections \ref{sec:iid} and \ref{sec:corr}, we consider a {\em deterministic} version of the Cournot game, i.e., one without production uncertainty. Understanding of this deterministic game provides context for our results; further, our proofs use the deterministic setting as a building block.

In the deterministic setting, we ignore the second stage of the game. Therefore the payoff for firm $i$ is:
\begin{equation}\label{eqn:opi}
\opi_i (x_i,{\bd x}_{-i}) = p\left(\sum_{l=1}^N x_l\right) x_i.
\end{equation}
Compared with \eqref{eqn:pi}, note that the cost for shortfall is omitted.
For the rest of the paper, we use overlined variables to represent quantities in the deterministic game.

Consider the game defined by $(\opi_1,\dots,\opi_N)$. By the same reasoning as in Proposition \ref{prop:NE}, a unique Nash equilibrium exists for this game, denoted by $(\ox_1,\dots,\ox_N)$.  The social welfare $U(y)$ is maximized at $\ymax$. The efficiency ratios for this game are defined as:
\begin{equation} \label{eqn:rbar_utility}
\overline r_W = \frac{ U \left(\sum_i^N \ox_i \right)}{U(\ymax)}
\end{equation}
and
\begin{equation} \label{eqn:rbar}
\overline{r}_O = \frac{\sum_{i=1}^N {\ox_i}}{\ymax}.
\end{equation}
%Since there is no uncertainty, we measure efficiency with respect to $\ymax$ and analogously define the efficiency ratio $\overline{r}$ as

The behavior of $\overline{r}_W$ and $\overline{r}_O$ as $N$ increases is well understood; see, e.g., \cite{Johari05}. As noted in Proposition \ref{prop:rbar}, $\overline{r}_W$ approaches $1$ and the game becomes efficient in the limit of many firms. As we show in Section \ref{sec:iid}, this is no longer true if production uncertainty is present.
\begin{proposition}[Corollary 18 in \cite{Johari05}] \label{prop:rbar}
\begin{equation*}
\lim_{N \rightarrow \infty} \overline{r}_W \rightarrow 1 \mbox{ and } \lim_{N \rightarrow \infty} \overline{r}_O \rightarrow 1.
\end{equation*}
\end{proposition}

\subsection{Coalitions}

In this section, we define Cournot competition among coalitions of firms. Given $N$ firms, let $\mS_1,\dots,\mS_K$ be a partition of $\{1,\dots,N\}$. Let $(x_1,\dots,x_N)$ be a vector of production levels for each firm. The aggregate production commitment of group $\mS_k$ is denoted as $x(\mS_k)$:
\begin{equation*}
x(\mS_k)=\sum_{i \in \mS_k} x_i.
\end{equation*}
Similarly let $X (\mS_k)=\sum_{i \in \mS_k} X_i$ denote the aggregate (random) realized capacity of the group $\mS_K$. The payoff of the group $\mS_k$ is defined as:
\begin{equation}
\pi_k (x(\mS_k)) = p \left(\sum_k x(\mS_k)\right) x(\mS_k) - \E[\left(x(\mS_k ) - X (\mS_k)\right)^+].
\end{equation}
Note that, as for the system operator, a coalition benefits by being able to use the excess production of one member to offset the shortfall of another.  Thus the penalty incurred by the coalition is the shortfall between their {\em aggregate} realized capacity and {\em aggregate} production commitment.   Note that we do not consider the internal profit sharing contracts of each coalition; instead, we focus on profit maximization of the coalition.

Given $\mS_1,\dots,\mS_k$, we can define a Cournot game among coalitions through the payoff functions $(\pi_1,\dots,\pi_k)$.  In this game the action for group $\mS_k$ is the aggregate production commitment $x(\mS_k)$; as with profit, we do not focus on how this commitment is divided among the individual firms.
The game played by coalitions is a ``scaled'' version of the original game played by $N$ individual firms. The key difference is that the penalty is \emph{not linear} in the firms. By an analogous result to Lemma \ref{lem:aggregate},
\begin{equation}
\E[(x(\mS_k ) - X (\mS_k))^+] \leq \sum_{i \in \mS_k} \E[(x_i-X_i)^+].
\end{equation}
It is this reduction in risk that makes coalitions useful, as we describe in the subsequent sections.

\section{Independent Firms} \label{sec:iid}
%% !TEX root=IEEE_TAC.tex
In this section, we consider the efficiency of the two-stage Cournot game when the production uncertainty is i.i.d. across firms.  Since we are interested in the large $N$ regime, we need to specify how the random variables $(X_1, \ldots, X_N)$ scale as $N$ increases. Recall that $X_i$ models the realized capacity of production.  Since we hold the price function constant as we increase $N$, we should reasonably expect that each firm will produce an infinitesimal amount in the limit.  If we do not adjust the production capacity accordingly, then each firm will effectively face no production uncertainty in the large $N$ limit.

Formally, we adjust the scale of the production capacity of each firm according to the following assumption.
\begin{assumption} \label{assump:N}
Let $X$ be a symmetric continuous random variable with $\E[|X|^3] < \infty$. Let $\E[X]=\mu$ and we assume $\mu > \ymax$; let $\mbox{Var(X)}=\sigma^2$. Suppose there are $N$ firms. The random variables $X_1,\dots,X_N$ are drawn i.i.d.~according to the distribution of $X/N$.
\end{assumption}
Under the above assumption, the expected total capacity is fixed at $\mu$ and is divided evenly among the $N$ firms. The technical assumption of a bounded third moment avoids random variables with very heavy tails, and is largely made for analytical convenience.  The assumption $\mu > \ymax$ streamlines the proofs, but is not essential. \newrj{We note that Assumption \ref{assump:N} can be equivalently formulated by holding the production capacity of each firm constant, but instead scaling the price functional as $N \cdot p$ as the number of producers $N$ increases.}

Under Assumption \ref{assump:N}, all firms are {\em ex ante} identical: they have the same profit and face the same production uncertainty.  In this section we consider coalitional competition where the firms are divided evenly into $K$ groups; since we are interested in large $N$ scenarios, we assume, without loss of generality, $N$ is a multiple of $K$. The two extreme values of $K$ are $K=1$ and $K=N$: the former corresponds to a grand coalition, while the latter corresponds to competition among individual firms.

The next theorem is the main result of this section, which relates the efficiency ratios to the group size $K$.
\begin{theorem} \label{thm:rate}
Suppose there are $N$ firms and $X_1,\dots,X_N$ satisfies Assumption \ref{assump:N}. Let $(\mS_1,\dots,\mS_K)$ be $K$ groups where each group has $N/K$ firms. Let $(x(\mS_1), \dots,x(\mS_K))$ be the solution to the game $(\pi_1,\dots,\pi_K)$. Then the efficiency ratios scales as:
\begin{align}
 r_W &= \frac{ U \left( \sum_{k=1}^K x(\mS_k) \right) - \sum_{k=1}^K \E\left[ \left(x(\mS_k)-X(\mS_k)\right)^+ \right] }{U(\ymax')-\E \left[\left( \ymax'-\sum_i^N X_i \right)^+\right]} \nonumber \\
 & = 1- O\left(\frac{1}{K^2}\right)-O\left(\frac{K}{N}\right)  \label{eqn:rate_utility}
\end{align}
and
\begin{equation} \label{eqn:rate}
r_O =\frac{\sum_{k=1}^K x(\mS_k)}{\ymax'}=1-O\left(\frac{1}{K}\right) - O\left(\frac{K}{N}\right).
\end{equation}
\end{theorem}

We focus on the discussion of the results here; the full proof of Theorem \ref{thm:rate} is in Appendix \ref{app:thmrate}. The last two terms in \eqref{eqn:rate_utility} and \eqref{eqn:rate} can be interpreted as the effects of market power and production uncertainty, respectively. In \eqref{eqn:rate_utility}, the inefficiency due to market power scales as $1/K^2$, and it decreases as $K$ grows. On the other hand, the inefficiency due to production uncertainty scales as $K/N$, which decreases as $N/K$ (the number of members in each coalition) grows. Similarly, in \eqref{eqn:rate}, $1/K$ represents the inefficiency due to market power and $K/N$ represents the inefficiency due to uncertainty.  Note from \eqref{eqn:rate_utility} and \eqref{eqn:rate} that $r_W$ and $r_O$ approach $1$ as long as $K$ and $N/K$ both grow without bound. The following corollary gives the optimal coalition size for maximizing the rates at which they approach $1$.
\begin{corollary} \label{cor:opt_size}
The optimal coalition structure to maximize the right hand side of \eqref{eqn:rate_utility} is to divide $N$ firms into $\Omega(N^{1/3})$ groups, each of size $\Omega(N^{2/3})$.

The optimal coalition structure to maximize the right hand side of \eqref{eqn:rate} is to divide $N$ firms into $\Omega(\sqrt{N})$ groups, each of size $\Omega(\sqrt{N})$.
\end{corollary}
This Corollary follows directly from balancing the terms in Theorem \ref{thm:rate}. It is interesting to note that $r_W$ and $r_O$ are optimized by different coalition formations. However, both rates suggest that it is more efficient to have intermediate regimes of coalition sizes other than the two extremes of individual firms and the grand coalition.

The result in Theorem \ref{thm:rate} can be extended to a more general shortfall penalty, as in the following corollary.
\begin{corollary} \label{cor:rate}
Let $f$ be a  convex increasing function with bounded derivative, satisfying $f(x)=0$ for all $x \leq 0$. For a group $\mS$, let its total profit be given by:
\begin{equation}
\pi_\mS (x_\mS)=p\left(\sum_{i=1}^N x_i\right)\left(\sum_{k \in \mS} x_k\right)- \E[f(x_\mS-X_\mS)].
\end{equation}
Under the same conditions as Theorem \ref{thm:rate}, the efficiency ratios still scale as
\begin{equation*}
r_W= 1- O\left(\frac{1}{K^2}\right)-O\left(\frac{K}{N}\right)
\end{equation*}
and
\begin{equation*}
r_O= 1- O\left(\frac{1}{K}\right)-O\left(\frac{K}{N}\right).
\end{equation*}
\end{corollary}
The proof of this corollary is given in Appendix \ref{app:corrate}.

\begin{remark}
The scaling rates in Theorem \ref{thm:rate} and Corollary \ref{cor:rate} are stated as lower bounds.  \newrj{There are instances that demonstrate that these rates are tight. For example, let $p(x)=1-x$.  In this case the inefficiency due to market power can be calculated exactly and scales as $r_W = \Omega(1/K^2)$  and $r_O = \Omega(1/K)$. Let $X_i$ be a continuous random variable that satisfies Chebyshev's inequality with equality (see, e.g.,~\cite{Durrett99}). Then both $r_W$ and $r_O$ scale as $\Omega(K/N)$.  (See the proof of Theorem \ref{thm:rate} for details.)}
\end{remark}

It is worthwhile to note that the scaling rates in Theorem \ref{thm:rate} and Corollary \ref{cor:rate} represent the asymptotic behavior of firms. Constants of the terms in the scaling rate are determined by the particular distributions of production uncertainty, and the price function.  We illustrate the behavior of the efficiency ratio at finite $N$ with the following example.  \\

\begin{example}
\label{ex:1}
Let $p(y)=1-y$. Let $X$ be normally distributed as $\mathcal N(1.1,1)$. Note that $\ymax=1 < 1.1 $. Let $X_i$ be drawn i.i.d.~according to the distribution of $X/N$. Figure \ref{fig:rate} plots the efficiency ratio for two groups sizes: $\sqrt{N}$ and $N^{2/3}$. Theorem \ref{thm:rate} and Corollary \ref{cor:opt_size} show that groups of size $N^{2/3}$ is optimal for welfare efficiency and groups of size $\sqrt{N}$ is optimal for rate efficiency. Figures \ref{fig:rate_welfare} and \ref{fig:rate_x} valid these claims, although there is a switch over point in Fig. \ref{fig:rate_x}.
\begin{figure}[ht]
\centering
\begin{subfigure}{\linewidth}
\includegraphics[scale=0.4]{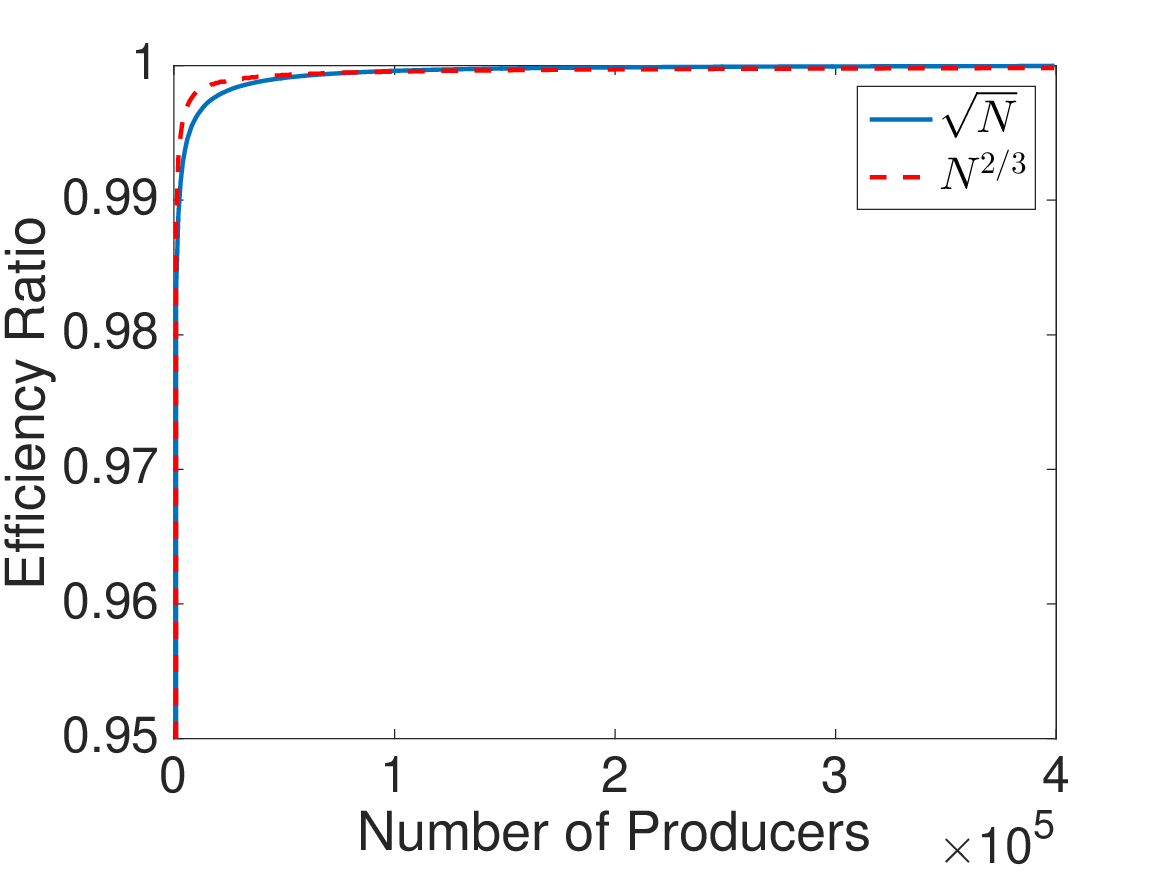}
\caption{Scaling of the welfare efficiency ratio.}
\label{fig:rate_welfare}
\end{subfigure}

\begin{subfigure}{\linewidth}
\includegraphics[scale=0.4]{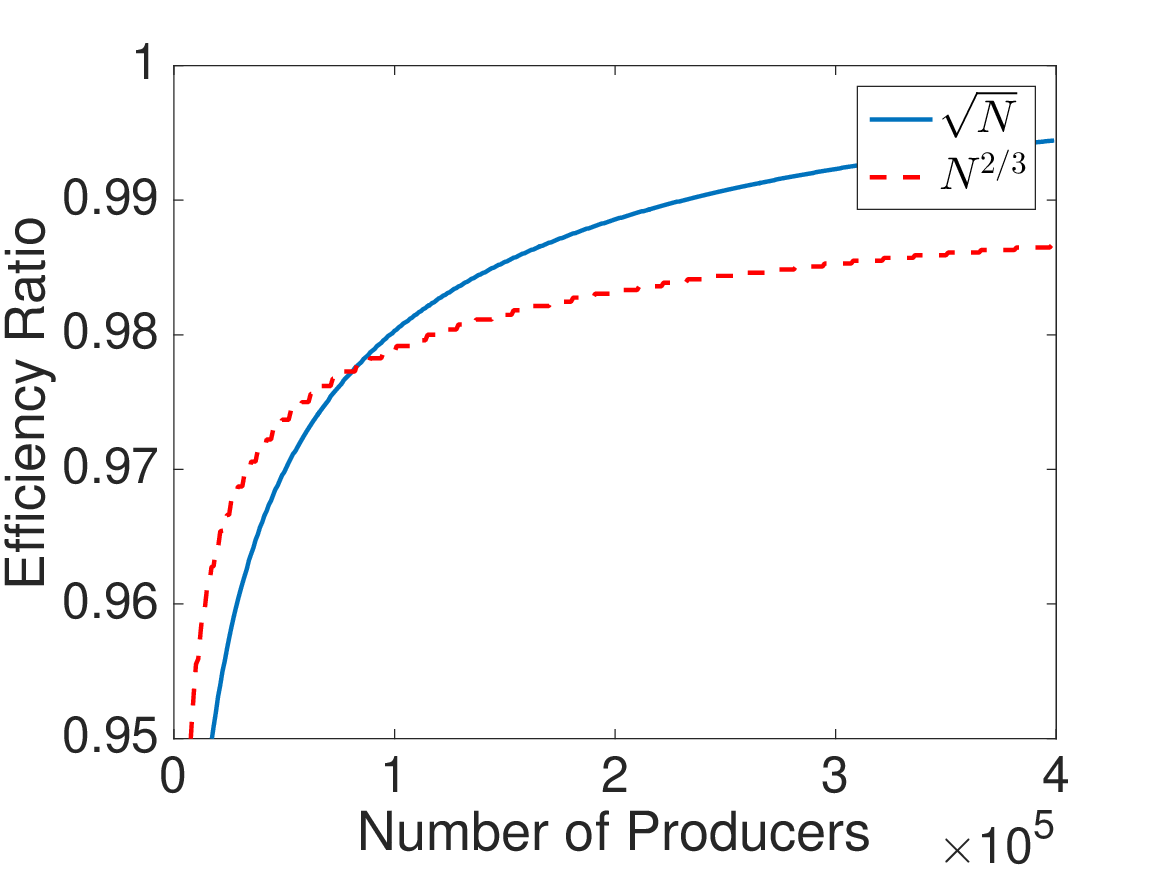}
\caption{Scaling of the quantity efficiency ratio.}
\label{fig:rate_x}
\end{subfigure}
\caption{Scaling of the efficiency ratios for groups of size $N^{2/3}$ and $\sqrt{N}$ (for $N^{1/3}$ and $\sqrt{N}$ total groups, respectively). The ratio of welfares, $r_W$, approaches $1$ faster when groups are of size $N^{2/3}$ than when groups are of size $\sqrt{N}$. The ratio of quantities, $r_O$, approaches 1 asymptotically faster when groups are of size $\sqrt{N}$ than when groups are of size $N^{2/3}$.}
\label{fig:rate}
\end{figure}
\end{example}
%
%\begin{example}
%\label{ex:2}
%In this example we show that the regime switching phenomenon in Figure \ref{fig:rate_log} may not occur. Here we consider a uniform random variable and the effects of competition and averaging are not as distinct as in Example 1.
%
%Again let $p(y)=1-y$. Let $X$ be a uniform random variable distributed as $Unif[0,2.2]$, with mean $\mu=1.1$. Let $X_i$ be drawn i.i.d. according to $X/N$. Figure \ref{fig:rate_unif} plots the efficiency ratio for groups of sizes $\sqrt{N}$ and $N^{2/3}$. Unlike in Figure \ref{fig:rate}, the efficiency ratio is always higher for groups of size $\sqrt{N}$. Also, a regime change is not present in the semi-log plot in Figure \ref{fig:rate_unif_log}.
%\begin{figure}[ht]
%\centering
%\includegraphics[scale=0.3]{rate_unif.eps}
%\caption{Efficiency ratio for groups of size $\sqrt{N}$ and $N^{2/3}$.}
%\label{fig:rate_unif}
%\end{figure}
%\begin{figure}[ht]
%\centering
%\includegraphics[scale=0.3]{rate_unif_log.eps}
%\caption{Semi-log plot of the efficiency ratio for groups of size $N^{2/3}$  and $\sqrt{N}$ (for $N^{1/3}$ and $\sqrt{N}$ total groups, respectively). In contrast to Figure \ref{fig:rate_log}, there are not regime changes. }
%\label{fig:rate_unif_log}
%\end{figure}
%\end{example}
%%\begin{figure}[ht]
%%\centering
%%\includegraphics[scale=0.3]{rate_unif_log.eps}
%%\caption{Semi-log plot of the efficiency ratio for groups of size $\sqrt{N}$ and $N^{2/3}$. In contrast to Figure \ref{fig:rate_log}, there are not regime changes. }
%%\label{fig:rate_unif_log}
%%\end{figure}
%%\end{example}

\section{Correlated Firms} \label{sec:corr}
%% !TEX root=IEEE_TAC.tex
In this section, we consider two models where firms have {\em correlated} production uncertainty.

\subsection{Weakly Correlated Firms}

The $O(K/N)$ terms in \eqref{eqn:rate_utility} and \eqref{eqn:rate} result from the law of large numbers.  The following corollary recovers the same result, using a version of the law of large numbers for correlated random variables.
\begin{corollary} \label{cor:weak}
Let $X$ be a random variable satisfying Assumption \ref{assump:X}. Let $\E[X]=\mu > \ymax$. Suppose there are $N$ firms. Assume the random variables $X_1,\dots,X_N$ each have marginal distribution that is the same as $X/N$. Let $(\mS_1,\dots,\mS_K)$ be $K$ groups where each group has $N/K$ firms. Let $(x(\mS_1), \dots,x(\mS_K))$ be the solution to the game $(\pi_1,\dots,\pi_K)$.

If
\begin{equation} \label{eqn:weak_cor}
\sum_{j=1}^N \left|\E\left[\mbox{cor}\left(X_i, X_j \right) \right]\right| \leq \frac{c}{N} \; \forall \; i
\end{equation}
%\sum_{j=1}^N \left|\E\left[\left(X_i-\frac{\mu}{N}\right)\left(X_j-\frac{\mu}{N}\right)\right]\right| \leq \frac{c}{N} \; \forall \; i
for some $c$ independent of $N$, then the efficiency ratios scale as in \eqref{eqn:rate_utility} and \eqref{eqn:rate}.
\end{corollary}
An example of correlated $X_i$'s satisfying the above condition is where $\Cov(X_i,X_j) \leq A \rho^{|i-j|}$, for some finite $A$ and $\rho < 1$.  This type of model captures a Hotelling-like geographic structure, where firms with similar indices are more likely to face the same production constraints.  It is particularly relevant in electricity markets, where wind turbines located near each other exhibit this behavior. The proof of Corollary \ref{cor:weak} is given in Appendix \ref{app:weak}.

\subsection{Strongly Correlated Firms}

Earlier, we considered firms with {\em weakly} correlated production capacity, in the sense that the correlations between firms decays as the number of firms grows.
In this section, we consider the case of \emph{strongly} correlated production capacities, where the correlation between all firms remains positive as $N$ grows.

When the firms have correlated capacities, results similar to Theorem \ref{thm:rate} are difficult to obtain in general since the limiting distribution of $\sum_i X_i$ does not necessarily concentrate; any such result will depend on the particular joint distribution of the $X_i$'s. For this section, we assume that the correlation between random variables arises from an {\em additive} model, as described in the following assumption.
\begin{assumption} \label{assump:Z}
Let $X$ be a zero mean continuous random variable with symmetrical density and satisfies $\E[|X|^3] < \infty$. The random variables $\hat{X}_1,\dots,\hat{X}_N$ are drawn i.i.d.~according to the same distribution as $X/N$. Let $Z$ be a continuous random variable with mean $\mu$, finite variance, and symmetrical density around its mean. The random variables $Z$ and $\hat{X}_1,\dots,\hat{X}_N$ are independent.  The random variable $X_i$ is given by:
\begin{equation}
X_i=\hat X_i + \frac{Z}{N}\ \text{for all}\ i.
\end{equation}
\end{assumption}

%Before discussing the efficiency of the Cournot game under correlated firms, we need to modify the social planner's problem.
Since $X_i$'s are strongly correlated, $\sum_i X_i$ no longer concentrates around its mean. Therefore, even for a system operator that jointly controls output of all firms, there is some residual uncertainty in the system. However, it is still beneficial to the system to share the production of all firms and then face the cost of the aggregate shortfall.
%The social planner's problem remains the same as in \eqref{eqn:U_max_y}, where it tries to maximize
%\begin{equation*}
%U (y) - \E\left[\left(y - \sum_{i=1}^N X_i \right)^+ \right].
%\end{equation*}
With the notations used in Assumption \ref{assump:Z}, the system operators's problem is
\begin{equation} \label{eqn:y_Z}
\max_{y \geq 0} U(y)- \E[(y-Z-\sum_{i=1}^N \hat X_i)^+].
\end{equation}
As before, let $\ymax'$ be the unique solution to \eqref{eqn:y_Z} and $\ymax' \leq \ymax$.

%In the limit of large $N$, therefore, we can view the social planner's optimization problem as the following:

%The second term in the objective function is motivated by the fact $\sum_i X_i \rightarrow \mu+Z$ as $N \rightarrow \infty$. As before, let $\ymax'$ be the unique solution to \eqref{eqn:y_Z}. Because of the additional uncertainty, $\ymax' \leq \ymax$.

%Again we assume that $N$ firms are divided evenly into $K$ groups. Let $(x(\mS_1), \dots, x(\mS_K))$ be the solution to the two-stage game $(\pi_1,\dots,\pi_K)$. Define the efficiency ratio $r$ as $\sum_k x(\mS_k)/y_{\max}'$. Strictly speaking, the mathematical rigorous limiting process should first calculate the ratio between the Nash equilibria and the social optimal outcome for a fixed $N$, then take $N$ to infinity. However, in our case, both quantities are real, nonzero, and positive numbers. Therefore the limit of the ratio equal to the ratio of the limits and we focus on measuring the Nash equilibria to $\ymax'$.

 Theorem \ref{thm:rate_Z} states that the efficiency ratios between the welfares and the quantities have the same large $N$ asymptotic behavior as in Theorem \ref{thm:rate}.
\begin{theorem} \label{thm:rate_Z}
Let $\hat{X}_1,\dots,\hat{X}_N$ and $Z$ be random variables that satisfy Assumption \ref{assump:Z}.  Let $(\mS_1,\dots,\mS_K)$ be a partition of $(1,\dots,N)$ with size $N/K$ each. Let $(x(\mS_1),\dots,x (\mS_k))$ be the solution to the two-stage game $(\pi_1,\dots,\pi_K)$. Suppose $\mu > \ymax'$. The efficiency ratios scale as:
\begin{align}
 r_W &= \frac{ U \left( \sum_{k=1}^K x(\mS_k) \right) - \sum_{k=1}^K \E\left[ \left(x(\mS_k)-X(\mS_k)\right)^+ \right] }{U(\ymax')-\E \left[\left( \ymax'-\sum_i^N X_i \right)^+\right]} \nonumber \\
 &= 1- O(\frac{1}{K^2})-O(\frac{K}{N}) \label{eqn:rate_Z_utility}
\end{align}
and
\begin{equation} \label{eqn:rate_Z}
r_O=\frac{\sum_k x(\mS_k)}{y_{\max}'}=1-O\left(\frac{1}{K}\right)-O\left(\frac{K}{N}\right).
\end{equation}
\end{theorem}
The proof of Theorem \ref{thm:rate_Z} proceeds in similar steps to the proof of Theorem \ref{thm:rate} and can be found in the Appendix \ref{app:rate_Z}.

\subsection{Simulation Results}

Here we plot the efficiency ratio for correlated firms and compare it to independent firms. Similar to Example \ref{ex:1}, let $p(y)=1-y$. Let $\hat{X}$ be normally distributed as $\mathcal{N} (1.1, 0.71)$ and let $Z$ be normally distributed as $\mathcal{N}(0,0.71)$; note that with this definition, the variance of $\hat{X}+Z$ is 1. Let $\hat{X}_i$ be drawn i.i.d.~according to the same distribution as $\hat{X}/N$. Figure \ref{fig:scaling_cor} shows the efficiency ratio $r_O$ for groups of size $\sqrt{N}$ on a semi-log plot. As a baseline, we also plot the efficiency ratio where the random variables are drawn i.i.d.~with normal distribution $\mathcal{N} (1.1,1)$.
\begin{figure}[ht]
\centering
\includegraphics[scale=0.3]{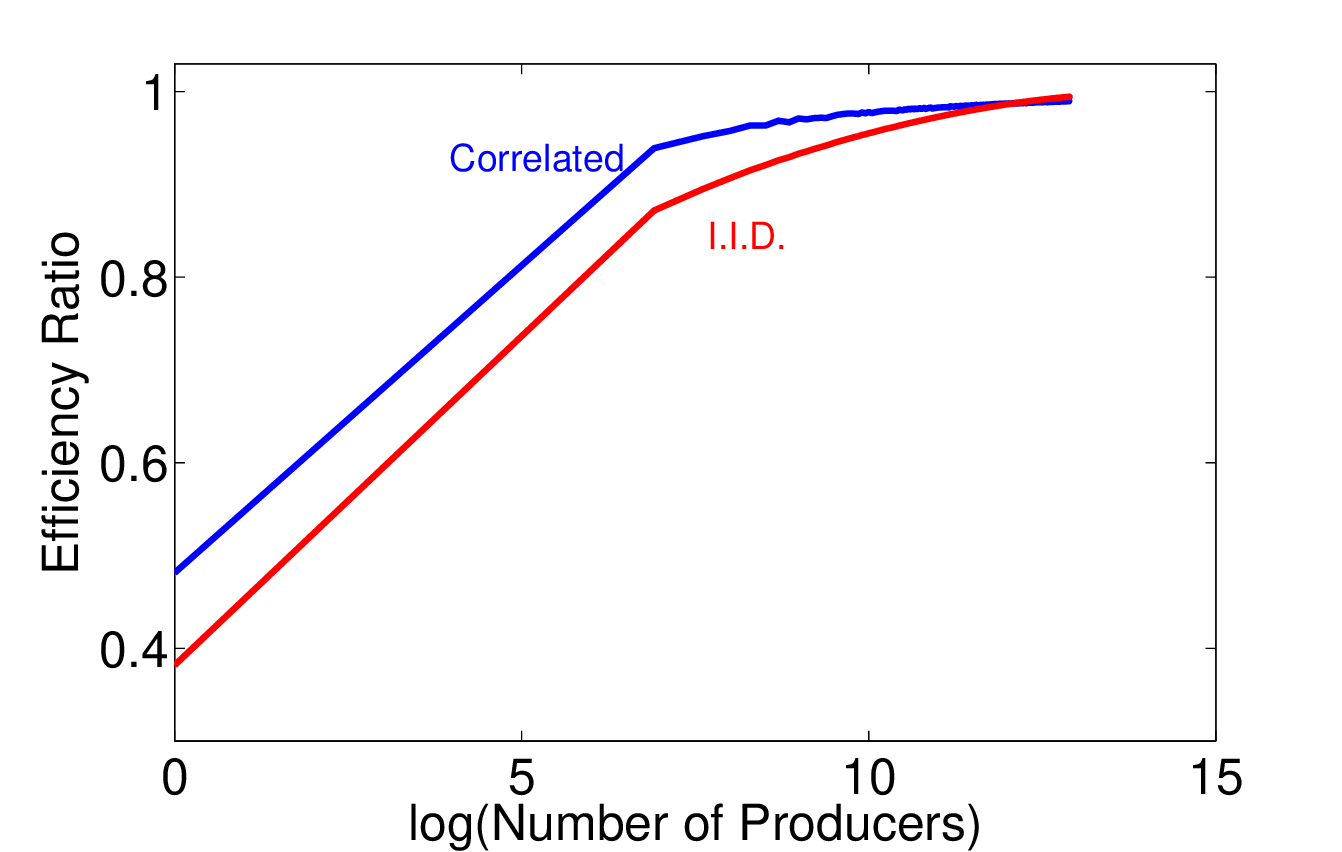}
\caption{Efficiency ratio for correlated firms and i.i.d.~firms as a function of the log of the number of firms. The groups size are $\sqrt{N}$ for both cases. We see that in the correlated case, the efficiency ratio increases much faster than the i.i.d.~case.}
\label{fig:scaling_cor}
\end{figure}

From Figure \ref{fig:scaling_cor}, we can see that the efficiency ratio $r_O$ approaches $1$ much faster if the firms are correlated. This is not unexpected, since production uncertainty is dominated by the common random variable $Z$, and the individual randomness can be averaged out more easily.
%Another way to interpret this result is that since the production uncertainty also reduces the social planner's optimal welfare, the efficiency ratio is relatively higher.

\section{Cournot Oligopsony} \label{sec:oligopsony}
%% !TEX root=IEEE_TAC.tex
There is a natural Cournot oligopsony game corresponding to the Cournot oligopoly game discussed in previous sections. Instead of suppliers of a common good, firms can be thought as consumers of a common good. If there are uncertainties in the \emph{demands} of the firms, a natural question arises: does coalition increase the efficiency of Cournot oligopsony as in the case of Cournot oligopoly?

%then we face the same problem as in Cournot oligopoly: does coalition increase the efficiency of the game?

We answer this question by showing that under certain demand side models, all results from previous sections can be directly applied through an equivalency between Cournot oligopoly and Cournot oligopsony. Before the technical details, let us consider the following motivating example of parking in business districts.
\begin{example}
Many cities around the world are experiencing increasing vehicle traffic in downtown and business areas. Currently, most cities allocate a fixed number of parking spaces to a firm.\footnote{Excluding private parking garages} For firm $i$, suppose it is allocated $x_i$ number of parking spaces. The amount of vehicles (demand of the firm) that visit the firm is a random variable denoted by $X_i$. If a vehicle successfully finds parking, the firm derives a certain amount of utility. If parking is not available when users try to visit a firm, it gets no utility. Therefore firm $i$'s utility is proportional to $\min(x_i,X_i)$.

An increase of parking spaces is correlated with the total traffic into downtown areas. This increase in traffic could potentially increase congestion, and therefore parking is tightly controlled by the city. We model this cost with a price function $\hat{p}(\sum x_i)$ and firm $i$ has cost $x_i \hat{p}(\sum x_i)$.
\end{example}

%\begin{example}
%Seattle, like many cities around the world, faces increasing vehicle traffic in its downtown areas. Currently, the city allocates a fixed number of parking spots or a fixed parking area to its downtown firms (businesses).\footnote{This does not include private parking garages.} For firm $i$, suppose its allocated $x_i$ parking spots. Let $X_i$ denote the \emph{random variable} that represents the amount of vehicles that actually visits firm $i$. For each vehicle that successfully parks, the firm derives a certain amount of utility. On the other hand, if a vehicle arrives and parking is full, the firm gets no utility. Therefore we suppose that firm $i$'s utility is proportional to $\min(x_i,X_i)$.
%
%An increase of parking spaces is correlated with the total traffic into downtown areas. This increase in traffic could potentially increase congestion, and therefore parking is tightly controlled by the city. We model this cost with a price function $\hat{p}(\sum x_i)$ and firm $i$ has cost $x_i \hat{p}(\sum x_i)$.
%\end{example}

With the above motivation, we define the expected payoff for firm $i$ to be
\begin{equation}\label{eqn:pi_demand}
T_i(x_i,\bd{x_{-i}})=\E[\min(x_i,X_i)]-x_i \hat{p}\left(\sum_{l=1}^N x_l \right).
\end{equation}
Each firm is price anticipatory and firm $i$ chooses $x_i>0$ to maximize $T_i$ for a given $\bd{x_{-i}}$. A Nash equilibrium for this game is defined to be the same as in \eqref{eqn:nash}. As the analogous of Assumption \ref{assump:p}, we make the following assumption on the price $\hat{p}$.
\begin{assumption}
We assume that:
\begin{enumerate} \label{assump:p_demand}
\item $\hat{p}$ is strictly increasing and $0<\hat{p}(0)<1$;
\item  $\hat{p}(y)$ is convex and differentiable on $y \geq 0$ with $\hat{p}'(0^+)>0$;
\item $\hat{p}(y) \rightarrow  \infty$ as $y \rightarrow \infty $.
\end{enumerate}
\end{assumption}
The $\hat{p}(0)<1$ assumption is made to avoid trivial solutions, since if $\hat{p}(0)>1$ no firm will choose a positive bid. Note we assume that the utility and the price function, $\E[\min(x_i,X_i)]$ and $\hat{p}$, have been appropriately scaled to have the same units.

It is straight forward to show that a Nash equilibrium exists for the game $(T_1,\dots,T_N)$ as given in Proposition \ref{prop:NE_demand}.
\begin{proposition} \label{prop:NE_demand}
Suppose $\hat{p}$ satisfies Assumption \ref{assump:p_demand} and the $X_i$'s satisfy Assumption \ref{assump:X}. Then there exists a unique Nash equilibrium $\bd x$ for the game defined by $(T_1,\dots,T_N)$.
\end{proposition}
The proof of this proposition is given in Appendix \ref{app:NE_demand}.

\subsection{Social Welfare and Efficiency}
Given the price function $\hat{p}$, we define the aggregate cost as:
\begin{equation}
C(y)= \int_0^{y} \hat{p}(z) dz.
\end{equation}
Similar to the oligopoly case, a system operator with control of all the firms would always aggregate the supply and demand as shown by Lemma \ref{lem:aggregate_demand}.
\begin{lemma} \label{lem:aggregate_demand}
Suppose $x_1,\dots,x_N$ are a set of real numbers and $X_1,\dots,X_N$ are a set real random variables each satisfying Assumption \ref{assump:X}. Then
\begin{equation} \label{eqn:minxX}
\E\left[\min \left( \sum_{i=1}^N x_i,\sum_{i=1}^N X_i \right) \right] \geq \sum_{i=1}^N \E\left[\min(x_i,X_i)\right].
\end{equation}
\end{lemma}
The proof of this lemma is given in Appendix \ref{lem:aggregate_demand}.

Therefore, from a system operator's point of view, the optimal allocation is characterized by solving the following problem:
\begin{subequations} \label{eqn:social_demand}
\begin{align}
\mbox{maximize } & \E\left[\min \left( \sum_{i=1}^N x_i,\sum_{i=1}^N X_i \right) \right] - C\left(\sum_{i=1}^N x_i \right) \\
\mbox{subject to } & x_i\geq 0, \forall \; i.
\end{align}
\end{subequations}
%Similar to \eqref{eqn:U_max}, the social planner controls the aggregate assignment of all firms, and faces the aggregate uncertainty in $X_i$'s. Again, the social planner is able to use the demand of a firm to offset the shortfall of another firm.

The objective function in \eqref{eqn:social_demand} only depends on the sum of allocations, so we can write it as
\begin{subequations} \label{eqn:social_y}
\begin{align}
\mbox{maximize } & \E\left[\min \left( y,\sum_{i=1}^N X_i \right) \right] - C\left(y \right) \\
\mbox{subject to } & y\geq 0.
\end{align}
\end{subequations}

Let $\hat{y}_{\max}$ be the solution to \eqref{eqn:social_y}. As in Definition \ref{defn:r_u}, we define the two \emph{efficiency ratios}, one based on the welfares and one based on the quantities.
%\begin{definition}
%Consider the game $(T_1,\dots,T_N)$. Let $\mathcal X$ be the set of all Nash equilibria of the game. The efficiency ratio $r$ is defined as
%\begin{equation}
%r=\inf_{\bd x \in \mathcal X} \frac{\sum_{i=1}^N x_i } {\hat{y}_{\max}}.
%\end{equation}
%\end{definition}
\begin{definition} \label{defn:r_demand}
Consider the Cournot game $(T_1,\dots,T_N)$. Let $(x_1,\dots,x_N)$ be the Nash equilibrium of this game. The efficiency ratio $r_W$, with respect to the welfares, is:
\begin{equation} \label{eqn:r_u_demand}
r_W= \frac{\sum_{i=1}^N \E\left[\min \left( x_i, X_i \right) \right] - C\left(\sum_{i=1}^N x_i \right)}{\E\left[\min \left( \hat{y}_{\max},\sum_{i=1}^N X_i \right) \right] - C\left(\hymax \right)}.
\end{equation}
The efficiency ratio $r_O$, with respect to the quantities, is:
\begin{equation} \label{eqn:r_demand}
r_O=  \frac{\sum_{i=1}^N x_i } {\hat{y}_{\max}}.
\end{equation}
\end{definition}
As in Sections \ref{sec:iid} and \ref{sec:corr}, we are interested in the behaviors of $r_W$ and $r_O$ as $N$ grows. The next section shows that there is an exact equivalence between the game $(T_1,\dots,T_N)$ studied in this section and the game $(\pi_1,\dots,\pi_N)$ studied in Sections \ref{sec:iid} and \ref{sec:corr}. Consequently the main results in Theorems \ref{thm:rate} and \ref{thm:rate_Z} carry over directly.

\subsection{Equivalence between Oligopoly and Oligopsony}
A certain level of symmetry between oligopoly and oligopsony is commonly expected in Cournot games, but the key difficulty is to ensure that efficient allocations and Nash equilibria remain unchanged between the two models. The following theorem is the main result of this section.

\begin{theorem} \label{thm:oligopsony}
Suppose that Assumptions \ref{assump:X} and \ref{assump:p_demand} hold.
%All derivatives are taken with respect to $x_i$ (treating $X_i$ as a constant). For each firm $i$, define:
%\begin{equation}
%\hat{C}_i'(x_i,X_i)=
%1-\frac{d}{d x_i} \E\left[ \min(x_i,X_i) \right].
%\end{equation}
%Define the associated cost function $\hat{C}_i(x_i,X_i)=\int_0^{x_i} \hat{C}_i'(z,X_i) dz$.
Define a new price $p$ as:
\begin{equation}
p(y)=1-\hat{p}(y).
\end{equation}
Define an associated utility function $U(y)=\int_0^y p(z) dz$. Then
\begin{enumerate}
\item Assumption \ref{assump:p} is satisfied by $p$.
\item For any vector $\bd x\geq 0$, the following holds:
\begin{align}
& \E\left[\min\left(\sum_i x_i,\sum_i X_i \right)\right]-C\left(\sum_i x_i \right) \nonumber \\
& = U\left(\sum_i x_i\right)-\E\left[\left(\sum_i x_i-\sum_i X_i\right)^+\right), \label{eqn:eq_social}
\end{align}
as well as:
\begin{align}
T_i(x_i;\bd{x}_{-i})&=\E[\min(x_i,X_i)]-x_i \hat{p} \left(\sum_l x_l \right) \nonumber \\
&=x_i p\left(\sum_l x_l\right) - \E[(x_i-X_i)^+] \nonumber \\
&=\pi(x_i;\bd{x}_{-i}). \label{eqn:eq_i}
\end{align}
\item A vector $\bd x$ solves the system operator's problem in \eqref{eqn:U_max} if and only if it solves the system operator's problem in \eqref{eqn:social_demand}.
\item A vector $\bd x$ is a Nash equilibrium of the game defined by $(T_1,\dots,T_N)$ if and only if it is a Nash equilibrium of the game defined by $(\pi_1,\dots,\pi_N)$.
\end{enumerate}
\end{theorem}
In Appendix \ref{app:oligopsony} we provide proofs of the four claims above.

\section{Case Study of Wind Integration in Power Systems} \label{sec:power}
%% !TEX root=IEEE_TAC.tex
This section applies the result of this paper to electricity markets. This material is based on \cite{Zhang15}. We focus on the PJM control area in the United States (an area that covers most of the Mid-Atlantic states of the United States). In recent years, there has been a dramatic increase of wind generation in PJM. Among other effects, the one most pertinent to this paper is a significant drop in energy prices as wind penetration grows. Figure \ref{fig:pjm} (reproduced from \cite{Gil13}) shows the day-ahead electricity price can drop by as much as 50\% with just 10\% of wind penetration.
\begin{figure}[ht]
	\centering
	\includegraphics[scale=0.22]{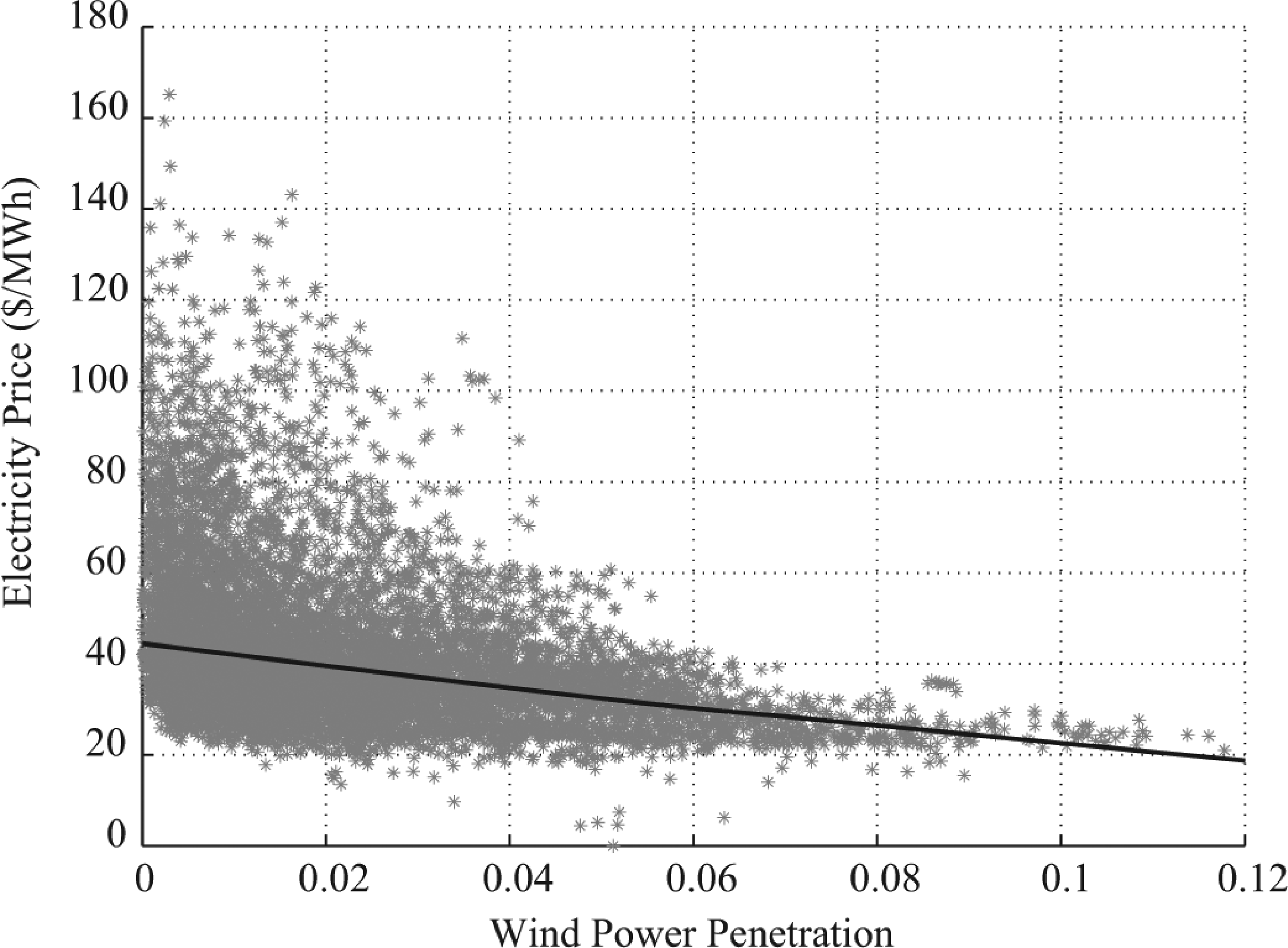}
	\caption{Scatter plot of PJM day-ahead prices and wind generation for 2012. Figure reproduced from \cite{Gil13}. The horizontal axis is the percentage penetration of wind, and the vertical axis is the average clearing price in the PJM area. As the amount of wind penetration increases from 0 to 10\%, the average price in the system drops by more than half.}
	\label{fig:pjm}
\end{figure}

The electricity market in PJM and most other regions of the United States operate with a two-stage structure. The first stage occurs roughly 24 hours before the actual time of delivery of electricity--called the day-ahead stage--where generators compete to satisfy the forecasted demand. The second stage is called the real-time stage since it adjusts supply and demand to meet any imbalance that may occur at the time of electricity delivery. Since most conventional generators (e.g., coal, nuclear, hydro, etc.) needs at least several hours to change production levels, most of the market is cleared at the day-ahead stage. In this paper, we are interested in the wind power producers and treat conventional generators as non-strategic entities. From PJM market reports (\cite{PJM13}), it seems that conventional generators already bid their true cost in the current market, and therefore would not change
their bid when wind power producers enter the market (bidding lower makes no economic sense and bidding higher decreases the chance that they are cleared). However, investigating the joint action between renewable and conventional generators is an important area of future research.

We think of wind power producers as the firms competing in the day-ahead market. Each firm faces uncertain supply, since its  output is determined by the wind conditions in the future; in addition, each firm can impact the price as shown in Fig. \ref{fig:pjm}. Therefore the Cournot game developed in this paper can be used to model their behavior.

For an empirical study, we use wind data produced by National Renewable Energy Lab (NREL) for Eastern United States \cite{nrel}. Both wind forecast value and forecast error are included in the data set. We interpret the forecast value as the mean of the random supply and the error as the variation in the supply. There are 302 wind farm locations that are in the PJM control area. Figure \ref{fig:correlation} shows the normalized standard deviation of the aggregate forecast as a function of the number of wind farms in the aggregate. Strong correlations can be observed between these forecasts.
\begin{figure}[ht]
\centering
\includegraphics[scale=0.5]{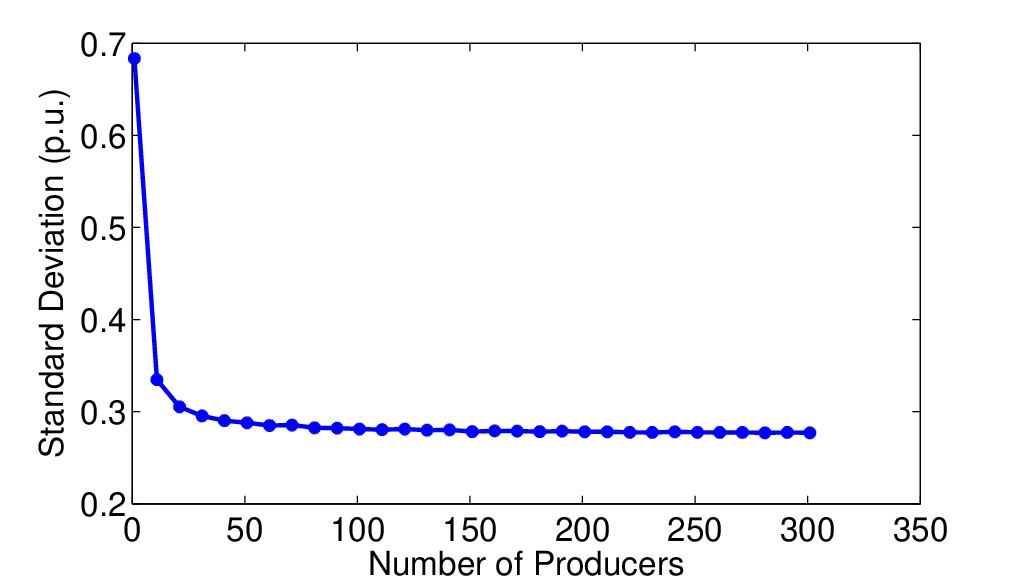}
\caption{Standard deviation as a function of the number of producers in an aggregation. The vertical axis is normalized by the total capacity of the aggregation.}
\label{fig:correlation}
\end{figure}

Applying the analysis in Section \ref{sec:corr}, Fig. \ref{fig:correlated} shows the optimal coalitions of the wind farms. Note the social optimal solution is taken to be the solution that maximizes the amount of wind injected into the market. Interested readers can find a much more detailed analysis in \cite{Zhang15}.
\begin{figure}[ht]
\centering
\includegraphics[scale=0.5]{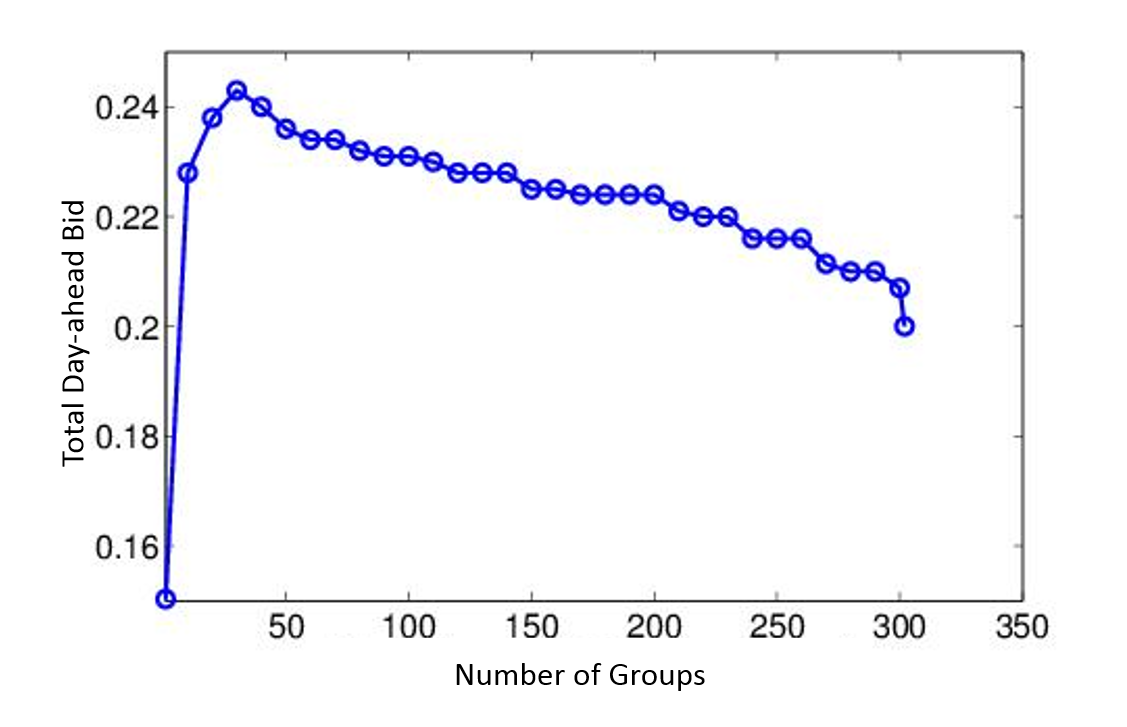}
\caption{Total amount of wind bid into the day-ahead market as a function of the number of groups for the wind farms in the NREL dataset. The maximum occurs at 30 groups, while the minimum occurs at a single group (grand coalition). }
\label{fig:correlated}
\end{figure}

\section{Conclusion} \label{sec:con}

In this paper we investigated strategic behavior of firms and coalitions in a Cournot game with production uncertainty.  We study the efficiency of Cournot competition by  characterizing a fundamental tradeoff: on one hand, market power increases as coalition size grows; on the other hand, the cost of production uncertainty is mitigated as coalition size grows.  We show there is a ``sweet spot'', in the sense
that there exist groups that are large enough to achieve the
uncertainty reduction of the grand coalition, but small enough
such that they have no significant market power.
%Namely, when there are $N$ firms, competition between groups of size $O(\sqrt{N})$ results in equilibria that are socially optimal.

These results have important implications for regulators in industries with production uncertainty, such as electricity markets.  In particular, our results suggest that within some limits, coalition formation among, e.g., generators of renewables may actually {\em increase} overall welfare.  We have validated these results in electricity markets in \cite{Zhang15}, where we empirically study the welfare benefits of coalitions of (a finite number of) wind power generators.

We conclude by noting two important open directions.  {\em First}, as previously noted, in many real markets (including electricity markets), the penalty for a production shortfall is not exogenous.  Rather, a firm may face a ``spot'' or secondary market, into which it can sell excess capacity, or from which it must buy additional capacity to cover a shortfall.  Modeling this two-stage market game remains an important challenge.   {\em Second}, all our results are asymptotic, though we do characterize the rate of convergence to efficiency with optimal coalition sizes.  With finitely many (potentially heterogeneous) firms, the regulator faces the potentially challenging problem of computing optimal coalitions as a benchmark.  Developing approaches for this problem remains an open issue.

%
% % Appendix
\appendices
%% !TEX root=IEEE_TAC.tex
\section{Proof of Proposition \ref{prop:NE}}
\label{app:NE}

  \begin{proof}
  We first observe that the strategy space of each firm can be restricted to a compact set, without loss of generality.
  For any vector $\xmi$ chosen by the other firms, if $\Pr(X_i < \ymax) >0$, then firm $i$ is always strictly better off by choosing 0 than choosing a quantity larger than $\ymax$. If $\Pr(X_i <\ymax)=0$ and $\sum_{k \neq i} x_k >0$, again firm $i$ is always strictly better off by choosing 0 than choosing a production larger than $\ymax$. If $\Pr(X_i <\ymax)=0$ and $\sum_{k \neq i} x_k = 0$, then firm $i$'s best strategy is to maximize $x_i p(x_i)$. By Assumption \ref{assump:p}, the unique maximizer of $x_i p(x_i)$ is strictly larger than 0 and strictly smaller than $\ymax$. Thus, we may restrict the strategy space of a firm to $[0,\ymax]$.

  Since $p$ satisfies Assumption \ref{assump:p}, $p(\sum_i x_i) x_i$ is concave. By Assumption \ref{assump:X}, $\E[(x_i-X_i)^+]$ exists and $-\E[(x_i-X_i)^+]$ is concave in $x_i$. By additivity of concave functions, $\pi_i$ is concave for all $x_i \geq 0$.

  The game defined by $(\pi_1,\dots,\pi_N)$ with strategy spaces $([0,\ymax],\dots,[0,\ymax])$ is now a strictly concave game: each payoff $\pi_i$ is continuous and strictly concave in $\bd x$ and the strategy space of firm $i$ is a nonempty compact set. By Rosen's existence theorem (see \cite{Rosen65}), there is a unique Nash equilibrium for this game.

  Furthermore, suppose $\sum x_i > \ymax$. Then at least one of the $x_i$'s is positive. But then firm $i$ is strictly better off if it reduces $x_i$. Therefore at equilibrium, $\sum x_i \leq \ymax$.
  \end{proof}

\section{Proof of Lemma \ref{lem:aggregate}}
\label{app:aggregate}

  \newbz{
\begin{proof}{$ $}
This lemma follows from Jensen's inequality. Let $W_1,\dots,W_N$ be $N$ random variables. Let $g(\cdot)=(\cdot)^+$. Since $g$ is convex, we have
\begin{equation} \label{eqn:Jensen_W}
  \E[g(\sum_i^N \frac{1}{N} W_i)] \leq \frac{1}{N} \sum_i^N \E[g(W_i)].
\end{equation}
Since $g$ is homogenous, multiplying both side of \eqref{eqn:Jensen_W} by $N$ gives
\begin{equation*}
  \E[g(\sum_i^N W_i)] \leq \sum_i^N \E[g(W_i)].
\end{equation*}
Identify $W_i$ with $x_i-X_i$ finishes the proof.
\end{proof}
}

% We first show that given any two real vectors $(x_1,\dots,x_N)$ and $(X_1,\dots,X_N)$, the following holds
% \begin{equation*}
% \left(\sum_{i=1}^N x_i - \sum_{i=1}^N X_i \right)^+ \leq  \sum_{i=1}^N  \left( x_i-X_i \right)^+.
% \end{equation*}
% Let $\mathcal{A}$ denote the set of indices $i$ for which $(x_i-X_i) \geq 0$. Then
% \begin{align*}
% \left(\sum_{i=1} x_i - \sum_{i=1} X_i \right)^+  &= \left(\sum_{i \in \mathcal{A}} x_i - \sum_{i \in \mathcal{A}} X_i + \sum_{i \in \bar{\mathcal{A}}} x_i - \sum_{i \in \bar{\mathcal{A}}} X_i \right)^+   \\
% & \stackrel{(a)}{\leq} \left(\sum_{i \in \mathcal{A}} x_i - \sum_{i \in \mathcal{A}} X_i \right)^+   \\
% & \stackrel{(b)}{=}   \sum_{i \in \mathcal{A}} (x_i-X_i)^+  \\
% & \stackrel{(c)}{=}   \sum_{i \in \mathcal{A}} (x_i-X_i)^+ + \sum_{i \in \bar{\mathcal{A}}} (x_i-X_i)^+  \\
% & = \sum_{i=1}^N   \left( x_i-X_i \right)^+  .
% \end{align*}
% The inequality denoted by $(a)$ follows from the fact that $(a+b)^+ \leq a^+$ if $b<0$; $(b)$ follows from the fact that every pair of $x_i-X_i$ is positive; $(c)$ follows from that the second term in the sum is $0$.
%
% Taking expectation of $(X_1,\dots,X_N)$ finishes the proof. \flushright $\blacksquare$
% \end{proof}
% \end{APPENDIX}

\section{Proof of Theorem \ref{thm:rate}}
\label{app:thmrate}

  \begin{proof}
  We fill first assume \eqref{eqn:rate} holds and use it to prove \eqref{eqn:rate_utility}. Then we prove \eqref{eqn:rate} to be true.

  First we observe that all groups are symmetric, therefore each of them has the same expected cost:
  \begin{align}
  \sum_{k=1}^K \E\left[ \left(x(\mS_k)-X(\mS_k)\right)^+ \right]&= K \E\left[ \left(x(\mS_k)-X(\mS_k)\right)^+ \right] \nonumber \\
  &= \E\left[ \left(K x(\mS_k)-K X(\mS_k)\right)^+ \right]. \nonumber
\end{align}
  For notational convenience, let $x_K= K x(\mS_k)$ and $X_K= K X(\mS_k)$. Note that $\E[X_K]=\mu$. Let $\hat X_K$=$X_K -\mu$ be the zero mean version. Similarly, we can write $\sum_i^N X_i = \mu+\hat X$, where $\hat X$ is zero mean.   Using these notations, the efficiency ratio $r_W$ is
  \begin{equation}
  r_W= \frac{U(x_K)-\E\left[ \left( x_K-\mu - \hat X_K \right)^+ \right]}{U(\ymax')-\E \left[ \left( \ymax'- \mu - \hat X \right)^+ \right]}.
  \end{equation}
  The following proposition is useful for bounding $r_W$.
  \begin{proposition} \label{prop:VW}
  Let $V$ and $W$ be independent continuous random variables with symmetric densities. Furthermore suppose $E[V]=E[W]=0$. Let $a$ be a positive real number. Then
  \begin{equation}
  \E \left[ \left( V+W-a\right)^+ \right] \geq \E \left[ \left( W-a\right)^+ \right].
  \end{equation}
  \end{proposition}
  The proof of Prop. \ref{prop:VW} is given later. Applying it, we have
  \begin{equation}
  \E\left[ \left( x_K-\mu - \hat X_K \right)^+ \right] \leq \E\left[ \left( x_K-\mu - \hat X- \hat X_K \right)^+ \right],
  \end{equation}
  where $X$ is independent of $\hat X_K$ and the expectation is now taken with respect to both $\hat X$ and $\hat X_K$.
  Therefore to lower bound $r_W$, it suffices to lower bound
  \begin{align}
  \Delta_U= &U(\ymax')-\E \left[ \left( \ymax'- \mu - \hat X \right)^+ \right]- U(x_K) \nonumber \\
  &-\E\left[ \left( x_K-\mu - \hat X -\hat X_K \right)^+ \right]. \label{eqn:r_diff_X}
\end{align}

  Inserting an independent copy of $\hat X$ into the expectation allows us to analyze \eqref{eqn:r_diff_X} using a Taylor expansion. Let $\delta=\ymax'-x_K$,
  \begin{align}
  U(x_K)&=U(\ymax'-\delta) \nonumber \\
  &\approx U(\ymax')- \delta U'(\ymax') + \frac{\delta^2}{2} U''(\ymax'), \label{eqn:U_taylor}
\end{align}
  where we neglect the higher order terms in the Taylor expansion. Let $g(y)=\E \left[ \left( y- \mu - \hat X \right)^+ \right]$. To obtain a Taylor expansion of $\E\left[ \left( x_K-\mu - \hat X -\hat X_K \right)^+ \right]$ around $\ymax'$, we use conditional expectation by successively conditioning on $\hat X_K$ and $\hat X$:
  \begin{align}
  & \E\left[ \left( x_K-\mu - \hat X -\hat X_K \right)^+ \right] \nonumber \\
  &= \E_{\hat{X}} \left[ \E_{\hat{X}_K} \left[ g(\ymax'-\delta-\hat X_K) | \hat{X} \right] \right] \nonumber \\
  &\approx \E_{\hat X} \left[ \E_{\hat{X}_K} \left[ g(\ymax') - (\delta+\hat X_K) g'(\ymax') + \frac{(\delta+\hat X_K)^2}{2} g''(\ymax') \right] \right] \nonumber \\
  &\stackrel{(a)}{=} g(\ymax') - \delta g'(\ymax') + \frac{\delta^2}{2} g''(\ymax')+\frac{\E[\hat X_K^2]}{2} g''(\ymax'), \label{eqn:E_taylor}
  \end{align}
  where again the higher order terms are neglected and $(a)$ follows from $\E[\hat X_K]=0$ and the independence between $\hat{X}$ and $\hat{X}_K$.

  Substituting \eqref{eqn:U_taylor} and \eqref{eqn:E_taylor} into \eqref{eqn:r_diff_X} yields:
  \begin{align}
  \Delta_U & \approx U(\ymax')- g(\ymax') - \left\{ U(\ymax')- \delta U'(\ymax') + \frac{\delta^2}{2} U''(\ymax') \right\} \nonumber \\
  & \quad + \left\{ g(\ymax') - \delta g'(\ymax') + \frac{\delta^2}{2} g''(\ymax')+\frac{\E[\hat X_K^2]}{2} g''(\ymax') \right \} \nonumber \\
  & = \delta \left(U'(\ymax')-g'(\ymax')\right) + \frac{\delta^2}{2} \left(g''(\ymax')-U''(\ymax') \right) \nonumber \\
  & \quad +\frac{\E[\hat X_K^2]}{2} g''(\ymax') \nonumber \\
  & \stackrel{(a)}{=} \frac{\delta^2}{2} \left(g''(\ymax')-U''(\ymax') \right) +\frac{\E[\hat X_K^2]}{2} g''(\ymax'), \label{eqn:g_U}
  \end{align}
  where $(a)$ follows from first order optimality conditions and fact that $\ymax'$ maximizes $U(y)-g(y)$. We now derive more explicit formulas for the constants in \eqref{eqn:g_U}:
  \begin{align*}
  g''(\ymax') &= \frac{d^2}{dy^2} \E[(y-\mu-\hat X)^+] |_{\ymax'} \\
  &= \frac{d}{dy} \E[ \mathds{1}(y-\mu-\hat X)^+] \\
  &= \frac{d}{dy} \Pr(\hat X \leq y-\mu) \\
  &= \hat f (y-\mu),
  \end{align*}
  where $\hat f$ is the pdf of $\hat X$; and $U''(\ymax') = p'(\ymax)$.

  By Assumption \ref{assump:X}, $\hat f$ exists; by Assumption \ref{assump:p}, $p'(\ymax')$ is negative. Therefore both $g''(\ymax')$ and $-U''(\ymax')$ are positive constants.

  The deviation $\delta^2$ can be calculated from \eqref{eqn:rate}. By \eqref{eqn:rate}, $\delta$ scales as $O(\frac{1}{K})+O(\frac{K}{N})$ and $\delta^2$ scales as
  \begin{equation*}
  O\left(\frac{1}{K^2} \right)+O \left(\frac{1}{N} \right)+O\left(\frac{K^2}{N^2}\right).
  \end{equation*}
  Since $\hat X_K$ is zero mean,
  \begin{align*}
  \E[\hat X_K^2]&=\mbox{Var}(X_K)=\mbox{Var}\left( K \sum_{i=1}^{N/K} X_i \right)  \\
  &= K^2 \frac{1}{N K} \sigma^2 =\frac{K}{N} \sigma^2=O\left(\frac{K}{N}\right).
\end{align*}

  Combining the above arguments, we have
  \begin{align*}
  & U(\ymax')-g(\ymax')-(U(x_K)-\E[(y-\mu-\hat{X}-\hat X_K)^+]) \\
   &= O\left(\frac{1}{K^2} \right)+O \left(\frac{1}{N} \right)+O\left(\frac{K^2}{N^2}\right)+ O\left(\frac{K}{N}\right) \\
  &= O\left(\frac{1}{K^2} \right)+ O\left(\frac{K}{N}\right).
  \end{align*}
  By Prop. \ref{prop:VW}, $\E[(y-\mu-\hat{X}-\hat X_K)^+] \geq \E[(y-\mu-\hat X_K)^+]$ and we have the efficiency ratio scales as
  \begin{equation}
  r_W=1- O\left(\frac{1}{K^2} \right)- O\left(\frac{K}{N}\right).
  \end{equation}

  We now prove \eqref{eqn:rate}. The strategy of the proof proceeds in two steps. First we consider a deterministic game with $K$ players, and  bound the difference between the Nash equilibrium of the game $(\pi_1,\dots,\pi_K)$ and the Nash equilibrium of the deterministic game.  Then we bound the difference between the latter and $\ymax$.

  Let $(\mS_1,\dots,\mS_K)$ be an equal sized partition of $(1,\dots,N)$, with each coalition of  size $N/K$. Consider a deterministic game defined by $(\opi_1,\dots,\opi_K)$ where
  \begin{equation}
  \opi_k = p\left(\sum_{m=1}^K \ox(\mS_m) \right) \ox(\mS_k).
  \end{equation}
   Let $(\ox(\mS_1),\dots,\ox (\mS_K))$ be a Nash equilibrium of this game. By symmetry, $\ox(\mS_k)$ are of the same value. We denote this value by $\ox_K$; it satisfies:
  \begin{equation} \label{eqn:ox}
  p(K \ox_K)+p'(K \ox_K) \ox_K=0.
  \end{equation}

  Let $(x(\mS_1),\dots,x(\mS_K))$ be a Nash equilibrium of the game $(\pi_1,\dots,\pi_K)$. Again, this equilibrium is symmetric. Denote the common production level of each coalition by $x_K$. Similarly, denote $X(\mS_k)$ by $X_K$.   Since $E[(x_K-X_K)^+]$ is increasing in $x_K$, we have $x_K \leq \ox_K$. Therefore we can rewrite $x_K$ as $\ox_K-\Del$, for some $\Del \geq 0$ that solves:
  \begin{align}
  & p(K (\ox_K-\Del))+p'(K (\ox_K-\Del)) (\ox_K-\Del) \nonumber \\
  & \quad -\Pr(X_K \leq \ox_K-\Del)=0. \label{eqn:xD}
\end{align}

  Subtracting \eqref{eqn:ox} from \eqref{eqn:xD}, we have
  \begin{align*}
  & [p(K\ox_K -K\Del)-p(K \ox_K)] + [p'(K\ox_K-K\Del)(\ox_K-\Del) \\
  & \quad -p'(K\ox_K)\ox_K]-\Pr(X_K\leq \ox_K-\Del) =0.
\end{align*}
  Since $p$ is concave and decreasing, $p'(K\ox_K)<p'(K\ox_K-K\Del) < 0$. Also since $\ox_K$ is positive, the term in the second set of brackets, $p'(K\ox_K-K\Del)(\ox_K-\Del)-p'(K\ox_K)\ox_K$ is greater than or equal to zero. Also since $\Del$ is positive, $\Pr(X_K \leq \ox_K-\Del) \leq \Pr(X_k \leq \ox_K)$. Therefore:
  \begin{equation} \label{eqn:1}
  p(K\ox_K -K\Del)-p(K \ox_K) -\Pr(X_K\leq \ox_K) \leq 0.
  \end{equation}
  Since $p$ is concave and decreasing,
  \begin{equation} \label{eqn:2}
  p(K\ox_K -K\Del)-p(K \ox_K) \geq K \Del (-p'(0)).
  \end{equation}
  Combining \eqref{eqn:1} and \eqref{eqn:2}, we obtain a bound on $K \Del$, or $\sum_k \ox(\mS_k) -\sum_k x(\mS_k)$, as
  \begin{equation} \label{eqn:Kdelta}
  K \Del \leq \frac{\Pr(X_K \leq \ox_K)}{-p'(0)} \leq \frac{\Pr(X_K \leq \ymax/K)}{-p'(0)}.
  \end{equation}
  By assumption, $\mu > y_{\max}$; applying Chebyshev's inequality gives $\Pr(X_K \leq \ymax/K) = O(K/N)$. Therefore $K \Del =O(K/N)$.

  Now we bound the gap between $K \ox_K$ and $y_{\max}$. Let $\delta=y_{\max}-K \ox_K$.  Substituting into \eqref{eqn:ox} gives:
  \begin{equation}
  p(\ymax-\delta)+p'(\ymax-\delta)\left(\frac{\ymax-\delta}{K}\right)=0.
  \end{equation}
  Since $p$ is decreasing and $\delta$ is positive,
  \begin{equation}
  p(\ymax-\delta)+p'(\ymax-\delta)\left(\frac{\ymax}{K}\right) \leq 0.
  \end{equation}
  Rearranging yields
  \begin{align}
  p(\ymax-\delta) & \leq -p'(\ymax-\delta)\left(\frac{\ymax}{K}\right) \nonumber \\
  & \stackrel{(*)}{\leq } - p'(\ymax)\left(\frac{\ymax}{K}\right) \label{eqn:pdelta_1},
  \end{align}
  where $(*)$ follows from the fact $p'$ is negative and decreasing ($p$ is decreasing and concave).
  Since $p$ is concave and decreasing, and $p(\ymax)=0$,
  \begin{equation} \label{eqn:pdelta_2}
  p(\ymax-\delta) \geq \delta \frac{p(0)}{\ymax}.
  \end{equation}
  See Figure \ref{fig:pdelta} for an illustration of \eqref{eqn:pdelta_2}. Combining \eqref{eqn:pdelta_1} and \eqref{eqn:pdelta_2} gives
  \begin{equation}
  \delta \frac{p(0)}{\ymax} \leq  -p'(\ymax)\frac{\ymax}{K},
  \end{equation}
  or
  \begin{equation}
  \delta \leq \frac{1}{K} \frac{-p'(\ymax) \ymax^2}{p(0)}.
  \end{equation}
  Therefore $\delta= O(1/K)$.

  \begin{figure}[ht]
  \centering
  \psfrag{p(0)}{$p(0)$}
  \psfrag{p(y)}{$p(\ymax-\delta)$}
  \psfrag{dp}{\large $\frac{\delta p(0)}{\ymax}$}
  \psfrag{ymax-d}{$\ymax-\delta$}
  \psfrag{ymax}{$\ymax$}
  \includegraphics[scale=0.3]{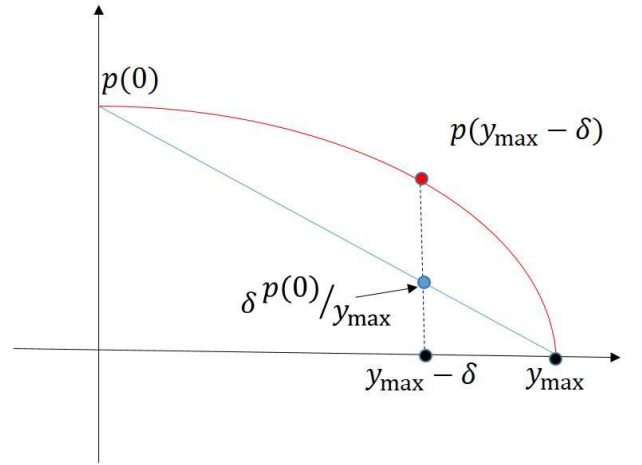}
  \caption{Graphical demonstration of \eqref{eqn:pdelta_1}.}
  \label{fig:pdelta}
  \end{figure}

  Combining the two parts of the proof,
  \begin{equation}
  \ymax-\sum_k x(\mS_k)=O\left(\frac{1}{K}\right)+O\left(\frac{K}{N}\right).
  \end{equation}
  Dividing both sides by $\ymax$ gives the desired result:
  \begin{equation}
  r_O=1-O\left(\frac{1}{K}\right)-O\left(\frac{K}{N}\right).
  \end{equation}

  Finally we prove Prop. \ref{prop:VW}. Let $f_V$ and $f_W$ be the pdf of $V$ and $W$, respectively. Then
  \begin{align*}
  & \E\left[ (V+W-a)^+ \right]-\E\left[ (W-a)^+ \right]\\
 &= \int_{0}^{\infty} \int_{\mu-v}^{\mu} (u+v-\mu) f_W(w) f_V(v) dw dv \\
  & + \int_{-\infty}^0 \int_{\mu}^{\mu+x} (\mu-w) f_W(w) f_V(v) dw dv \\
  & + \int_{\mu}^{\infty} \int_{\mu-w}^{w-\mu} v f_V(v) f_W(w) dv dw \\
  & + \int_{0}^{infty} \int_{\mu}^{\mu+x} w f_W(w) f_V(v) dw dv.
  \end{align*}
  By symmetry of $f_V$,
  \begin{equation*}
  \int_{\mu}^{\infty} \int_{\mu-w}^{w-\mu} v f_V(v) f_W(w) dv dw =0.
  \end{equation*}
  By symmetry of $f_W$,
  \begin{align*}
  &\int_{-\infty}^0 \int_{\mu}^{\mu+x} -w f_W(w) f_V(v) dw dv \\
  &+\int_{0}^{infty} \int_{\mu}^{\mu+x} w f_W(w) f_V(v) dw dv =0.
\end{align*}
  Therefore
  \begin{align*}
  & E\left[ (V+W-a)^+ \right]-\E\left[ (W-a)^+ \right] \\
   &= \int_{0}^{\infty} \int_{\mu-v}^{\mu} (u+v-\mu) f_W(w) f_V(v) dw dv \\
  & + \int_{-\infty}^0 \int_{\mu}^{\mu+x} \mu f_W(w) f_V(v) dw dv \\
  & \geq 0.
  \end{align*}
  \end{proof}

\section{Proof of Corollary \ref{cor:rate}}
\label{app:corrate}

  \begin{proof}{$ $}
  Let $(\mS_1,\dots,\mS_K)$ be a equal sized partition of $(1,\dots,N)$ (each of size $N/K$). Consider the deterministic game $(\opi_1,\dots,\opi_K)$ played by the groups. Let $(\ox(\mS_1),\dots,\ox (\mS_K))$ be a Nash equilibrium of this game. Let $(x(\mS_1),\dots,x(\mS_K))$ be a Nash equilibrium of the game $(\pi_1,\dots,\pi_K)$.   Since $E[f(x_K-X_K)]$ is increasing in $x (\mS_k)$, $x_K \leq \ox_K$, we can rewrite $x_K$ as $\ox_K-\Del$ with $\Del \geq 0$. Following the proof of Theorem \ref{thm:rate}, we obtain (in place of \eqref{eqn:Kdelta}):
  \begin{equation}
  K \Del \leq \frac{\E[f'(x_K-X_K)]}{-p'(0)}.
  \end{equation}
  Since $f'$ is bounded, $\E[f'(x_K-X_k)] \leq B \Pr(x_K-X_K \geq 0)$ for some $B$. Therefore $\Del$ scales as $O(K/N)$. Following the same steps as in the proof of Theorem \ref{thm:rate}, the efficiency ratio $r_O$ scales as $1-O\left(\frac{1}{K}\right)-O\left(\frac{K}{N}\right)$.  An analogous proposition to Prop. \ref{prop:VW} can be derived for a convex and increasing $f$, and a similar Taylor expansion argument can be used to bound $r_W$ based on $r_O$.
  \end{proof}

\section{Proof of Corollary \ref{cor:weak}}
\label{app:weak}

  \begin{proof}{$ $}
  Consider a group $\mS_k$. It suffices to show that $\Pr (\sum_{i \in \mS_k} X_i \leq \ymax/K)$ is $O(K/N)$ since the rest of the proof proceeds exactly as that of Theorem \ref{thm:rate}.

  Equation \eqref{eqn:weak_cor} implies that
  \begin{equation}
  \Pr \left(\sum_{i \in \mS_k} X_i \leq \ymax/K\right) \leq \frac{c}{(\mu-\ymax)^2} \frac{K}{N}=O\left(\frac{K}{N}\right).
  \end{equation}
  (See, e.g., \cite{DeGroot10}, for this standard result.)
  \end{proof}

\section{Proof of Theorem \ref{thm:rate_Z}}
\label{app:rate_Z}

\begin{proof}{$ $}
This proof follows a similar path to that of Theorem \ref{thm:rate}. We assume \eqref{eqn:rate_Z} to be true and prove \eqref{eqn:rate_Z_utility} holds. Then we prove \eqref{eqn:rate_Z}.

By symmetry, $x(\mS_k)$ are equal for all $k$ and we denote $\sum_{k=1}^K x(\mS_K)$ as $x_K$ and $x_K=K x(\mS_1)$. Let $X_K$ denote $K X(\mS_k)$. By Assumption \ref{assump:Z}, the random variable $X_K$ can be related to $Z$ as
\begin{equation*}
X_K= Z+ \hat{X}_K
\end{equation*}
where $\hat{X}_K$ has the same distribution as
\begin{equation*}
\hat{X}_K=K \sum_{i=1}^{N/K} \hat{X}_i.
\end{equation*}
 Therefore
\begin{align*}
& U \left( \sum_{k=1}^K x(\mS_k) \right) - \sum_{k=1}^K \E\left[ \left(x(\mS_k)-X(\mS_k)\right)^+ \right] \\
& = U \left( x_K \right) - \E\left[ \left( x_K-Z-\hat X_K \right)^+ \right].
\end{align*}
Let $\hat X=\sum_{i=1}^N \hat{X}_i$. Then
\begin{equation*}
r_W= \frac{U \left( x_K \right) - \E\left[ \left( x_K-Z-\hat X_K \right)^+ \right]}{U \left(  \ymax' \right) - \E\left[ \left( \ymax'-Z-\hat X \right)^+ \right]}.
\end{equation*}
Applying Prop. \ref{prop:VW}, $\E\left[ \left( x_K-Z-\hat X_K \right)^+ \right] \leq \E\left[ \left( x_K-Z-\hat X-\hat X_K \right)^+ \right]$ where $\hat X$ is independent of $\hat X_K$. By exactly the same argument as in the proof of Theorem \ref{thm:rate}, by assuming \eqref{eqn:rate_Z}, we can show
\begin{equation*}
r_W=1- O\left(\frac{1}{K^2} \right)- O\left(\frac{K}{N}\right).
\end{equation*}

 We prove \eqref{eqn:rate_Z} holds by first defining the following intermediate game. Define the payoff function:
\begin{equation} \label{eqn:opi_Z}
\opi_k=p\left(\sum_k x (\mS_k)\right) x (\mS_k) - \E[(x (\mS_k) - \frac{1}{K} (Z-\mu))^+].
\end{equation}
Let $(\ox(\mS_1),\dots,\ox (\mS_K))$ be a Nash equilibrium of the game defined by $(\opi_1,\dots,\opi_K)$ (given in \eqref{eqn:opi_Z}). By symmetry, $\ox(\mS_k)$ are the same for all $k$. We denote this value by $\ox_K$; it satisfies:
\begin{equation} \label{eqn:ox_Z}
p(K \ox_K)+p'(K \ox_K) \ox_K-\Pr(\ox_K \geq \frac{1}{K} (Z+\mu))=0.
\end{equation}

Let $(x(\mS_1),\dots,x(\mS_K))$ be a Nash equilibrium of the game $(\pi_1,\dots,\pi_K)$. Again by symmetry, all coalitions choose the same production level. Denote this value by $x_K$. Similarly, denote $X(\mS_k)$ as $X_K$.   Since $E[(x_K-X_K)^+]$ is increasing in $x (\mS_k)$, $x_K \leq \ox_K$. Therefore we can rewrite $x_K$ as $\ox_K-\Del$, for some $\Del \geq 0$ that solves:
\begin{align}
& p(K (\ox_K-\Del))+p'(K (\ox_K-\Del)) (\ox_K-\Del) \nonumber \\
& \quad -\Pr(\ox_K-\Del \geq X_K)=0. \label{eqn:xD_1}
\end{align}
By definition of $X_K$, \eqref{eqn:xD_1} can be written as:
\begin{align}
& p(K (\ox_K-\Del))+p'(K (\ox_K-\Del)) (\ox_K-\Del) \nonumber \\
& \quad -\Pr \left(\ox_K - \Del \geq \frac{1}{K} Z + \sum_{i=1}^{N/K} \hat{X}_i\right)=0. \label{eqn:xD_Z}
\end{align}
Subtracting \eqref{eqn:ox_Z} from \eqref{eqn:xD_Z} and following the steps of the proof in Theorem \ref{thm:rate}, we have
\begin{align}
K \Delta (-p'(0)) & \leq \Pr \left(\ox_K - \Del \geq \frac{1}{K} Z + \sum_{i=1}^{N/K} \hat{X}_i\right) \nonumber \\
& \quad -\Pr\left(\ox_K \geq \frac{1}{K} (Z+\mu)\right) \nonumber \\
          & =  \Pr \left(K \ox_K - K\Del \geq Z + K \sum_{i=1}^{N/K} \hat{X}_i\right) \nonumber \\
          & \quad -\Pr(K \ox_K \geq Z+\mu). \nonumber
\end{align}
Since $\Del\geq 0$,
\begin{align*}
K \Delta (-p'(0))&\leq \Pr \left(K \ox_K \geq Z + K \sum_{i=1}^{N/K} \hat{X}_i \right) \\
& \quad -\Pr(K \ox_K \geq Z+\mu).
\end{align*}
It is convenient to associate the mean $\mu$ with $Z$ rather than $\hat{X}_i$. Define $Z'=Z+\mu$ and $\hat{X}_K'=K \sum_{i=1}^{N/K} \hat{X}_i$. With this change of variables, we need to bound
\begin{equation*}
\Pr (K \ox_K \geq Z' + \hat{X}_K')-\Pr(K \ox_K \geq Z').
\end{equation*}
By conditional probability and the independence of $Z'$ and $\hat{X}_K'$,
\begin{align*}
& \Pr (K \ox_K \geq Z' + \hat{X}_K')-\Pr(K \ox_K \geq Z') \\
= & \Pr(K \ox_K \geq Z' + \hat{X}_K'|\hat{X}_K'\leq 0) \Pr (\hat{X}_K' \leq 0) \\
& + \Pr(K \ox_K \geq Z' + \hat{X}_K'|\hat{X}_K'> 0) \Pr (\hat{X}_K' > 0) \\
& - \Pr(K \ox_K \geq Z') \\
= & \{\Pr(K \ox_K \geq Z') \\
& +\Pr(K \ox_K < Z', K \ox_K \geq Z' + \hat{X}_K'|\hat{X}_K'\leq 0) \} \Pr (\hat{X}_K' \leq 0) \\
&+ \Pr(K \ox_K \geq Z' + \hat{X}_K'|\hat{X}_K'> 0) \Pr (\hat{X}_K' > 0) \\
 & -  \Pr(K \ox_K \geq Z') \\
\leq &  \{\Pr(K \ox_K \geq Z') \\
& +\Pr(K \ox_K < Z', K \ox_K \geq Z' + \hat{X}_K'|\hat{X}_K'\leq 0) \} \Pr (\hat{X}_K' \leq 0) \\
& + \Pr(K \ox_K \geq Z') \Pr (\hat{X}_K' > 0) -  \Pr(K \ox_K \geq Z') \\
=& \Pr(K \ox_K < Z', K \ox_K \geq Z' + \hat{X}_K'|\hat{X}_K'\leq 0) \Pr (\hat{X}_K' \leq 0) \\
\leq & \Pr(K \ox_K < Z', K \ox_K \geq Z' + \hat{X}_K'|\hat{X}_K'\leq 0) \Pr (\hat{X}_K' \leq 0)\\
& + \Pr(K \ox_K < Z', K \ox_K \geq Z' \\
 & + \hat{X}_K'|\hat{X}_K'>0) \Pr (\hat{X}_K' > 0) \\
= & \Pr (K \ox_K < Z', K \ox_K \geq Z'+\hat{X}_K').
\end{align*}
By assumption $Z'$ is a continuous random variable with a bounded density function, and we denote its density by $f_{Z'}$. Let $f_{Z,\max}$ denote the maximum value of $f_{Z'}$. Also, denote the density of $\hat{X}_K'$ by $f_{X'}$. Then
\begin{align*}
&   \Pr (K \ox_K < Z', K \ox_K \geq Z'+\hat{X}_K') \\
 &= \int_{-\infty}^{0} \int_{K \ox_K}^{K \ox_K-x} f_{Z'} (z) f_{X'}(x) dz dx\\
  & \stackrel{(a)}{\leq } f_{Z,\max} \int_{-\infty}^{0} (-x) f_{X'} dx \\
  & \stackrel{(b)}{=} f_{Z,\max} \int_{0}^{\infty} x f_{X'} dx \\
  & \stackrel{(c)}{\leq} f_{Z,\max} \cdot \mbox{const} \cdot \mbox{var} (\hat{X}'),
\end{align*}
where $(a)$ follows from boundedness of $f_Z$, $(b)$ follows from the symmetry in $\hat{X}'$ and $(c)$ follows from the fact that $\hat{X}_K'$ has bounded variance. Therefore $\sum \ox_K -\sum x_K = O(\frac{K}{N})$.

Now we bound the difference between $\sum \ox_K$ and $\ymax'$. By definition $\ymax'$ solves
\begin{equation} \label{eqn:ymax}
p(\ymax')-\Pr(\ymax'-Z-\mu >0) =0.
\end{equation}
Write $K \ox_K$ as $\ymax'-\delta$ for some $\delta>0$, and subtracting \eqref{eqn:ymax} from \eqref{eqn:ox_Z} gives
\begin{align*}
& \{p(\ymax'-\delta)-p(\ymax')\} + p'(\ymax'-\delta)\frac{\ymax'-\delta}{K} \\
& \quad -\{\Pr(\ymax'-\delta-Z-\mu >0)-\Pr(\ymax'-Z-\mu >0)\}=0.
\end{align*}
Since $p$ is decreasing and concave, and $\delta >0$, we have $p'(\ymax'-\delta) (-\delta) >0$  $ p'(\ymax') \leq p'(\ymax'-\delta) \leq 0$, and $\{\Pr(\ymax'-\delta-Z-\mu >0)-\Pr(\ymax'-Z-\mu >0)\} <0$. Therefore
\begin{equation*}
p(\ymax'-\delta)-p(\ymax') \leq -p'(\ymax')\frac{\ymax'}{K}.
\end{equation*}
Since $p$ is concave and decreasing, $p(\ymax'-\delta)-p(\ymax') \geq \delta \frac{p(0)-p(\ymax')}{\ymax'}$ (see Figure \ref{fig:pdelta}). Thus:
\begin{equation*}
\delta  \leq -p'(\ymax')\frac{\ymax'^2}{K (p(0)-p(\ymax'))}.
\end{equation*}
Combining with the first half of the proof gives the desired result.
\end{proof}

\section{Proof of Proposition \ref{prop:NE_demand}}
\label{app:NE_demand}

  \begin{proof}{$ $}
  By Assumption \ref{assump:X}, $\E[\min(x_i,X_i)]$ exists and is concave (expectation over point-wise minimums). By Assumption \ref{assump:p_demand}, $x_i \hat{p}(\sum x_i)$ is convex. Thus $T_i$ is concave for all $x_i >0$.

  Since $\E[\min(x_i,X_i)]$ increasing at most linearly with $x_i$ and $\hat{p}(\sum x_i)$ is convex and increasing without bound, there exists constants $B_i$ such that $\E[\min(B_i,X_i)]-B_i \hat{p}(B_i) <0$. For firm $i$, there is no incentive to choose $x_i$ greater than $B_i$ since it would be strictly better off choosing $x_i=0$. Therefore we may restrict the strategy space of firm $i$ to $[0,B_i]$. The game defined by $(T_1,\dots,T_N)$ with strategy space $([0,B_1],\dots,[0,B_N])$ is now a symmetric strictly-concave game: each payoff $T_i$ is concave and continuous in $\bd x$ and the strategy spaces are nonempty and compact. Therefore a unique Nash equilibrium exists.
  \end{proof}

\section{Proof of Lemma \ref{lem:aggregate_demand}}
\label{app:aggregate_demand}

  \begin{proof}{$ $}
  It suffices to prove that for any set of real numbers $a,b,A,B$ that
  \begin{equation} \label{eqn:minAB}
  \min(a+b,A+B)\geq \min(a,A)+\min(b,B).
  \end{equation}
  Lemma \ref{lem:aggregate_demand} follows from repeated application of \eqref{eqn:minAB}.

  To show \eqref{eqn:minAB} holds, we consider several cases. Let $\mathcal{A}$ be the condition that $\{a \leq A, b \leq B\}$ or $\{a \geq A, b \geq B\}$. Under $\mathcal{A}$,
  \begin{equation} \label{eqn:A}
  \min(a+b,A+B)= \min(a,A)+\min(b,B).
  \end{equation}
  The complement of $\mathcal{A}$ consists of two disjoint cases: $\{a>A, b < B\}$ and $\{a < A, b >B \}$
  If $\{a >A, b<B\}$,
  \begin{equation*}
  \min(a+b,A+B)=a+b > A+b =\min(a,A)+\min(b,B).
  \end{equation*}
  Similarly, if $\{a < A, b >B \}$
  \begin{equation} \label{eqn:barA}
  \min(a+b,A+B) > \min(a,A)+\min(b,B)
  \end{equation}
  under condition $\bar{\mathcal{A}}$. Combining \eqref{eqn:A} and \eqref{eqn:barA} yields \eqref{eqn:minAB}.
  \end{proof}

\section{Proof of Theorem \ref{thm:oligopsony}}
\label{app:oligopsony}
\begin{proof}{$ $}
We prove the four claims one by one.

\textit{Proof of Claim 1.} By Assumption \ref{assump:p_demand}, $\hat{p}$ is continuous, differentiable, convex and increasing, then $p$ is continuous, differentiable, concave and decreasing. Since $\hat{p}(y) \rightarrow \infty$, ${p} \rightarrow -\infty$; and since $\hat{p}(0)<1, \hat{p}'(0)>0$, ${p}(0)>0$ and ${p}'(0) <0$.

\textit{Proof of Claim 2.} By definition of $U$,
\begin{align}
U\left(\sum_i x_i \right) &= \int_0^{\sum_i x_i} p(z)dz \nonumber \\
& = \int_0^{\sum_i x_i} \left( 1-\hat{p}(z) \right) dz \nonumber \\
&= \sum_i x_i - C\left(\sum_i x_i \right).  \label{eqn:U_C}
\end{align}
Also
\begin{align}
& \sum_i x_i - E\left[\left(\sum_i x_i-\sum_i X_i\right)^+\right] \nonumber \\
 & = \E\left[\sum_i x_i -\left(\sum_i x_i-\sum_i X_i\right)^+\right] \nonumber \\
&= \begin{cases}
\E \left[ \sum_i x_i - \sum_i x_i + \sum_i X_i \right] &\mbox{ if } \sum_i x_i > \sum_i X_i \\
\E \left[ \sum_i x_i \right] &\mbox{ if }  \sum_i x_i \leq \sum_i X_i
\end{cases} \nonumber \\
&= \E\left[\min\left(\sum_i x_i, \sum_i X_i \right) \right]. \label{eqn:min}
\end{align}
Combining \eqref{eqn:U_C} and \eqref{eqn:min} yields \eqref{eqn:eq_social}. The derivation of \eqref{eqn:eq_i} follows in similar fashion as well.

\textit{Proof of Claim 3.} We already established \eqref{eqn:eq_social} holds. It remains to prove that any solution $\bd x$ to \eqref{eqn:social_demand} would satisfy $\sum_i x_i < \ymax$, where $\ymax$ is the unique solution to $p(y)=0$. By our definition, $p(y)=1-\hat{p}(y)$; therefore $\ymax$ is the unique solution to $\hat{p}(\ymax)=1$. Now suppose $\hat{p}(\sum_i x_i) >1$, then looking at \eqref{eqn:social_y},
\begin{equation}
\frac{d}{d y} \left.\left(\E\left[min(y,\sum_i X_i)\right]\right)\right\rvert_{y=\sum_i x_i} \leq 1
\end{equation} and
\begin{equation}
\frac{d}{d y} C(y)\rvert_{y=\sum_i x_i} = \hat{p}(\sum_i x_i) > 1.
\end{equation}
Thus $\bd x$ is not an optimal solution to \eqref{eqn:social_demand}, or equivalently, any solution of \eqref{eqn:social_demand} must satisfy $\sum_i x_i < \ymax$.

\textit{Proof of Claim 4.} The proof of this claim proceeds very similarly to the proof of Claim 3. Since \eqref{eqn:eq_i} holds, we only need to show that in both the game $(T_1,\dots,T_N)$ and the game $(\pi_1,\dots,\pi_N)$, no firm would bid $x_i > \ymax$ in a Nash equilibrium. Following the previous proof, it is straightforward to show that if a firm chooses $x_i > \ymax$, then it would have negative payoff regardless of the actions of the other firms. Thus no firm would choose such an $x_i$.
\end{proof}

%
% \end{APPENDICES}

% History dates
%\received{February 2007}{March 2009}{June 2009}

% Electronic Appendix
%\elecappendix

%\medskip
% Bibliography
\bibliographystyle{IEEEtran} % outcomment this and next line in Case 1
\bibliography{mybib}

\end{document}